\documentclass[12pt]{article}    
\usepackage[margin=.75in]{geometry}
\usepackage{setspace}

\usepackage{dsfont}
\usepackage{setspace}

\usepackage{algorithm}
\usepackage{algpseudocode}
\usepackage{tikz}
\usepackage{subcaption}
\usepackage{multirow}
\usepackage{amsmath}
\usepackage{amssymb}
\usepackage{makecell}
\usepackage{booktabs} 
\usepackage{mathtools}
\usepackage{amsthm}

\newtheorem{proposition}{Proposition}
\newtheorem{lemma}{Lemma}
\newtheorem{definition}{Definition}

\usepackage[
  backend=biber,
  natbib=true,   
  style=apa,
  maxbibnames=99
]{biblatex}
\DeclareLabeldate{
  \field{date}
  \field{year}
  \field{eventyear}
  \field{origyear}
  \literal{nodate}
}

\begin{document}
\title{The Illusion of Collusion}
\author{Connor Douglas, Foster Provost, Arun Sundararajan\\
\textit{NYU Stern School of Business\footnote{cpd8405@stern.nyu.edu, fprovost@stern.nyu.edu, digitalarun@nyu.edu. We thank seminar participants at New York University, the 2024 International Conference on Information Systems and the EC'25 Workshop on Online Learning and Economics for feedback on earlier versions of this work.}}}
\date{February 2026}
\maketitle
\abstract{
\begin{spacing}{1.15} 
Algorithmic agents are used in a variety of competitive decision-making settings, including pricing contexts that range from online retail to residential home rental. We study the emergence of algorithmic collusion when competing agents employ multi-armed bandit algorithms and competition is modeled as a repeated Prisoner's Dilemma game. Notably, agents in our setting perform online learning with no prior model of game structure and have no direct knowledge of competitor states or actions---thus they cannot learn strategies that depend on these factors. These context-free bandits nonetheless frequently learn seemingly collusive behavior---a phenomenon we term \textit{naive collusion}. Our results reveal that whether naive collusion emerges depends starkly on the choice of behavior policy employed by bandit learners. The mechanism underpinning the emergence of collusive outcomes is \textit{synchronicity} in agent action plays, where synchronicity captures how often agents play the same action. We show that in the long-run, naive algorithmic collusion \textit{never} emerges when both agents use a broad class of persistently random algorithms, including the epsilon-greedy algorithm without epsilon decay, \textit{sometimes} emerges when both agents use greedy-in-the-limit algorithms which feature randomness during exploration but are asymptotically deterministic, and \textit{always} emerges when both agents use deterministic bandit learning algorithms like those in the well-known upper confidence bound (UCB) family. We highlight market and algorithmic conditions under which one can and cannot predict \textit{a priori} whether collusion will occur. Our findings have several policy implications: preventing pricing algorithms from conditioning their actions on competitor prices may not preclude algorithmic collusion, symmetry in algorithms may increase collusion potential, and the emergence of algorithmic collusion is path dependent. 
\end{spacing}
}%

\section{Introduction}

Artificial intelligence (AI) agents that make autonomous, data-driven decisions are now widespread. For example, AI agents are often used for pricing, in contexts ranging from setting Amazon product prices \citep{chen2016amazon} to determining residential real estate rental rates \citep{bortolotti2023algorithmic}. These agents can behave in a manner that suggests they are learning to collude in ways that may be perceived as anticompetitive and welfare-reducing \citep{4220818}.  For example, in a setting where sellers compete by setting prices, agents of competing sellers may choose to set the same high price.  

The possibility of such collusive behavior has led to regulatory scrutiny of AI pricing agents globally~\citep{harrington2018developing, Bichler2025Algorithmic}. In the United States, establishing that seemingly collusive behavior violates antitrust law often requires some evidence of intentional coordination or an ``exchange of wills,'' and thus, algorithmic collusion arising from independent optimization by competing agents challenges current enforcement strategies. Indeed, even by auditing an agent's data and learning algorithm, it may not be possible to determine if the associated algorithm has learned to price supracompetitively \citep{hartline25}.  Granted, prior academic research has established that when competing sellers independently use machine learning algorithms, supracompetitive prices can nevertheless emerge, but the economic and informational characteristics of settings that lead to algorithms converging on collusive outcomes remain unclear. These questions become particularly interesting when agents themselves are unaware of the competitive characteristics of their setting---or even unaware that they are competing at all. For firms using pricing algorithms without the capacity for rigorous market analysis, as with small, third-party online sellers, pricing decisions may very well be made without considering the strategic response of opponents. 



This paper's central question is: when will competing algorithms converge to (seemingly) collusive outcomes, absent any information whatsoever about the strategic interaction in which they are engaged?  We call this \textit{naive algorithmic collusion}. We investigate this question using the Prisoner’s Dilemma as our setting for strategic interaction. We study the repeated play of this game between agents who use textbook \textit{bandit learning algorithms}, which do not take into account the choices and outcomes of their competitors (e.g., competing prices) but instead learn and act based solely on their own individual action and payoff histories. 

Multi-armed bandit (or simply \emph{bandit}) learning approaches are especially well-suited to answer our question because these are the canonical methods for solving \textit{online learning} problems. These decision-making problems involve choosing from a set of actions whose outcomes are uncertain. The choice of action generates another sample to learn about the action's reward distribution, but the choice incurs some penalty or reward. Exploration and exploitation of information must therefore be balanced. These algorithms often directly update their estimates of the value of each available action from each reward signal received, and then incorporate these new value estimates into the subsequent choice of action. Bandit algorithms are used frequently in practice because their ``textbook'' code is interpretable and widely available, because they need minimal parameterization, and because they are remarkably effective despite their simplicity \citep{Sutton1998}. 

Our central finding is that the emergence of naive algorithmic collusion is quite common but depends starkly on the degree of randomness in the learning policies incorporated by the algorithms.  Naive algorithmic collusion \textit{always} emerges when competing symmetric agents both use a deterministic bandit learning algorithm, but \textit{never} emerges in the long-run when both agents use a broad class of non-deterministic bandit learning algorithms, including the widely used epsilon-greedy algorithm without epsilon decay. 

Since competing agents in practice may not be symmetric, we provide both analytical and experimental results that explore asymmetries across these algorithms.  For example, we find that introducing a ``small'' amount of asymmetry or randomness into an otherwise deterministic algorithm---for example, asymmetry or randomness in the tie-breaking rule of the upper-confidence bound (UCB) bandit learning algorithm, another textbook model---may not suffice to prevent the persistent emergence of naive algorithmic collusion. 

Beyond establishing the importance of randomness in the agent's algorithm in shaping the emergence of seemingly collusive pricing, our results also highlight that, rather than depending exclusively on either symmetry or observability of competitor actions, a different dimension of the realized sequence of outcomes that we call \textit{synchronicity} in actions---how often the other agent matches an agent’s chosen action in the same round, conditional on the agent taking the given action---shapes whether what is perceived as collusion emerges. Put differently, the emergence of seemingly collusive behavior can be path-dependent rather than being determined entirely by the nature of the learning algorithm or its symmetric use across agents. 

Importantly, all of these results are obtained in the (naive) setting wherein neither agent has any information about the game, their competitor's actions, or their competitor's outcomes.  Naive algorithmic collusion is thus not predicated on the discovery by the algorithms of any sort of complex repeated-game strategy like those frequently underpinning ``folk theorem” results that sustain collusion via mutual awareness of the threat of what is often an elaborate scheme of future punishments for deviation \citep{fudenbergmaskin1986}.

Our results can inform a growing number of lawsuits suggesting that collusion may be occurring as a result of coordinated use across firms of the same pricing algorithms \citep{doj_realpage_2024}. Courts assessing whether observed seemingly coordinated outcomes across firms are anticompetitive sometimes look for additional circumstantial evidence of coordinated action called ``plus factors'' \citep[p.~396]{Kovacic2011PlusFactors}.  As explained by former FTC commissioner William Kovacic and his coauthors, ``Central to the operation of laws that aggressively punish collusion are the definition and proof of concerted action. Powerful consequences flow from whether price increases observed in the marketplace emerge from individual or collective initiative'' \citep[p.~394]{Kovacic2011PlusFactors}. Many of the “plus factors” used to determine whether observed parallel conduct results from an agreement—such as actions contrary to an agent’s self-interest, behavior that cannot be explained as rational absent concerted action, and evidence of efforts to establish regular communication—are explicitly and notably absent from our models.    
Thus, the emergence nevertheless of seemingly collusive behavior in our setting---wherein competing agents have no direct knowledge of each other's existence or choices and all of their actions clearly flow from ``individual initiative"---provides new insight into how seemingly collusive outcomes may be purely a consequence of specific features of the learning algorithms being employed by independent online learning agents, or due to the chance realization of synchronous choices made by independent actors. 

The rest of this paper is organized as follows. In Section 2, we introduce related work. In Section 3, we define the setting (the iterated Prisoner's Dilemma) and introduce the analytical lens we will use to study these learning processes. In Section 4, we highlight a broad class of \textit{persistently random} algorithms that will not collude in the limit. Importantly, however, we highlight that the use of these algorithms is sub-optimal from the firm's perspective. In Section 5, we analyze a class of \textit{greedy-in-the-limit} algorithms whose use leads to non-zero probabilities of seemingly collusive behavior. Interestingly, these algorithms are asymptotically optimal, unlike those presented in Section 4. In Section 6, we analyze a class of deterministic algorithms which will \textit{certainly} collude, with algorithms in this class also being asymptotically optimal in the single-agent setting. We conclude in Section 7 with a discussion of the results, draw out their policy implications, and summarize some robustness checks for other settings and situations that cannot be modeled analytically. 

\section{Related Work}

The literature on algorithmic pricing collusion is burgeoning, spurred by uptake in pricing algorithms \citep{chen2016amazon}. 
Much of this recent work focuses on \textit{tacit} collusion by algorithms, which arises via online, adaptive experimentation by multiple pricing agents operating in the same market \citep{Bichler2025Algorithmic}. Given the demonstrated efficacy of adaptive pricing algorithms, understanding the mechanisms by which seemingly collusive outcomes can arise from their use is important \citep{misra2019dynamic}. 

Such outcomes have been documented empirically. For example, \citet{learning_to_coordinate} demonstrate general learned collusive behavior in the gasoline market. 
While this study is not expressly algorithmic, the findings highlight the tangible existence of learning processes leading to ostensibly coordinated outcomes absent characteristics typical of explicit collusive schemes. \citet{assad2020algorithmic} study a similar setting, examining the effect of prices on German gasoline markets when algorithmic pricing systems are adopted. The authors find that in duopoly markets, prices increase when both firms employ pricing algorithms. \citet{chen2016amazon} scrape pricing data from Amazon Marketplace, where third party sellers are making their own pricing decisions in a competitive market. The authors examine correlations in time series to show that many sellers are using dynamic pricing algorithms and discuss potential market distortions that could emerge as a consequence. This choice by sellers is consistent with the literature on dynamic pricing, which shows that adaptive bandit-like algorithms can lead to substantial increases in profit \citep{misra2019dynamic}.

Several other papers have studied the tacit collusion potential of adaptive learning algorithms in a variety of settings. In a widely cited paper from this literature, \cite{calvano2020artificial} provide evidence of pricing agents learning to collude in a simulated online market. 
Specifically, the authors focus on a standard, time-discounted Q-learning model that prices based on the last \textit{k} rounds of play. Using a simple model of price competition with differentiated products and logit demand, the authors experimentally demonstrate that competing algorithmic agents can learn a reward-punishment strategy consistent with those suggested by ``folk theorems” from game theory \citep{fudenbergmaskin1986}.  

 A recent survey of economic and legal challenges associated with algorithmic pricing is presented by \citet{gautier2020ai}, who describe the broad concept of tacit algorithmic collusion. The authors focus on the fragility of past work suggesting the potential for algorithmic pricing collusion and review more recent work on learned price-fixing schemes. Likewise, \citet{veljanovski2022pricing} takes a critical view of this literature, suggesting that algorithmic collusion is not of legitimate antitrust concern and that existing EU laws are sufficiently adaptable to cover cases that may arise. \citet{wang_algorithms_2023} highlight that algorithmic pricing agents may lower rather than raise revenue when there are competitors following rule-based pricing in the market. Another skeptical stance is presented by \citet{Kang2022Raising}, who highlight that learning agents vary in their collusive potential according to symmetries in information and algorithm construction.

\citet{miklos-thal_ai_2024} take a more balanced stance, while also noting that specific types of algorithms are more likely to display collusive behavior. Our work further underscores a need to understand the nuances of these algorithms in use; for example, we show that even within a certain class of learning algorithm, the implementation details can affect whether seemingly collusive or competitive outcomes are reached. We deepen theoretical understanding of algorithmic collusion by providing analytical results that supplement what has been a largely empirical or simulation-based literature, while seeding a unified framework for thinking about these learning processes and showing that certain algorithms will certainly \textit{not }learn to collude. 

As a part of their broader analysis of algorithmic pricing in multi-agent settings, \citet{brown2023competition} also examine the possibility for collusion to arise and their key recommendation to policymakers is to prohibit firms from using pricing algorithms that condition their prices on those of their competitors. Such conditioning is also the mechanism by which \citet{calvano2019} argue that supracompetitive prices are sustained. However, we find that seemingly collusive behavior can occur even when algorithms do not condition their choices on competitor prices, which raises new questions about policy recommendations of this form. 

Other studies have examined learning pricing behavior without explicit conditioning on opponent prices. 
For example, \citet{banchio2023a} study reinforcement learning agents that play a repeated Prisoner’s Dilemma with agents learning without conditioning on opponent prices. Their key finding is that collusion can arise via an endogenous linkage in algorithmic behavior they term ``spontaneous coupling.” Studying the same class of algorithm, \citet{waltman2008} observe cooperative behavior among Q-learning agents leveraging the Bellman equation in their update rules while engaged in a repeated Cournot game \citep{waltman2008}. \citet{Dolgopolov} characterizes a generalization of a stateless Q-learning agent and analytically describes the emergent outcomes of these algorithms playing a repeated Prisoner’s Dilemma.

In contrast with prior work, we focus instead specifically on bandit learning agents.  Bandit algorithms are widely taught as the fundamental paradigm for decision making under uncertainty. Unlike Q-learning, the field of bandit learning focuses expressly on action sampling strategies for optimal choice, whether that be fast convergence to the best action or maximization of long-term rewards \citep{Sutton1998}. Bandit learning is used in a multitude of settings, including pricing \citep{trovo2015multi}, and dynamic pricing is a key application \citep{bouneffouf_2020}. 

Our work is the most similar in spirit to \citet{hansen21}, which models competing multi-armed bandits playing a repeated pricing game. They explore coordinated behavior in this game by UCB algorithms. Specifically, they focus their analysis on a version of the UCB algorithm tuned to their market setting. Nonetheless, they do show the emergence of naive collusion in their setting and show that this result is sensitive to noise in reward realizations. We broaden this study to a general repeated Prisoner's Dilemma game and also examine a range of classes of bandit learning algorithms. We also generalize a result of \citet{hansen21}, showing that \textit{all} deterministic bandit algorithms will learn to collude. On the other hand, we show that this emergence of seemingly collusive behavior by bandit agents is not universal across bandit algorithms but algorithm-specific: as we describe later in the paper, it depends on the degree of randomness in the bandit algorithm's learning policy, and further, may be path-dependent.

\section{Setting}
\subsection{The Prisoner's Dilemma}
We model competition between two firms using a repeated Prisoner's Dilemma, a well-known game featuring tradeoffs between cooperation and conflict. The Prisoner's Dilemma has been used to model a wide variety of situations ranging from arms races to price competition~\citep{axelrod1980effective}. In the game's canonical formulation, there are two players ($n=2$): \textit{Player 1} and \textit{Player 2}.\footnote{A generalization to $n$ players is discussed later.} Each player has two possible actions: the cooperative or \emph{collusive} action $H$ (think of this as setting the high price) and the selfish or \emph{competitive} action $L$ (setting the low price). Denote Player $i$'s action as $a_i$, where $a \in A = \{H,L\}$. 

The payoffs of the Prisoner's Dilemma game are typically characterized by four parameters. Without loss of generality and to keep notation simple, we normalize two of these parameters to 0 and 1. The resulting two parameters define the relative payoffs from each of the four outcomes. We refer to these parameters as $\beta$ and $\gamma$. The resulting payoff matrix is illustrated in Table \ref{tab:reduced_payoff}.

\begin{table}[h!]
    \centering
    \setlength{\extrarowheight}{2pt}
    \begin{tabular}{cc|c|c|}
      & \multicolumn{1}{c}{} & \multicolumn{2}{c}{Player 2}\\
      & \multicolumn{1}{c}{} & \multicolumn{1}{c}{$H$}  & \multicolumn{1}{c}{$L$} \\\cline{3-4}
      \multirow{2}*{Player 1}  & $H$ & $(\beta,\beta)$ & $(0,1)$ \\\cline{3-4}
      & $L$ & $(1,0)$ & $(\gamma,\gamma)$ \\\cline{3-4}
    \end{tabular}
    \caption{Payoffs in the Prisoner's Dilemma, with $1>\beta>\gamma>0$.}
    \label{tab:reduced_payoff}
\end{table}

It is well known that playing  $L$ is a dominant strategy in a one-shot Prisoner's Dilemma, and, thus, the unique Nash equilibrium is $(L,L)$.  However, both players can achieve strictly higher payoffs if the actions chosen are $(H,H)$, the collusive outcome. We represent an outcome as the vector $o$, where $o = (a_1 , a_2)$ with $o \in O =\{H,L\}^2$. Each outcome maps to a reward vector $r \in [0,1]^2$, according to the payoff matrix in Table \ref{tab:reduced_payoff}.\footnote{Often, $\beta$ is constrained to being greater than $\frac{1}{2}$ to ensure that alternating between defecting (playing $L$ while the opponent plays $H$) and being defected on does not result in a higher long-term payoff than sustained cooperation. We opt to model the full space of the Prisoner's Dilemma rather than imposing this constraint. While this may lead to alternating behavior in context-aware agents, this strategy is not directly learnable in naive, context-free agents.}


In what follows, we model players using bandit learning algorithms to play an iterated (repeated) Prisoner's Dilemma game, and thus, outcome and reward vectors are indexed by period $t$.  Our analytical results model an infinitely repeated Prisoner's Dilemma.  Our simulations model a finitely repeated Prisoner's Dilemma. We present the mapping of outcomes to reward vectors in Table \ref{tab:state_index}. 

Notably, in our model, \emph{players are unaware of this game structure or even the existence of a game}.  Their decision making (and learning) only takes into account the payoffs associated with their own actions. Before describing this interplay in more depth, we first briefly describe aspects of bandit learning relevant to what follows. 

\subsection{Bandit Learning}
The problem of the multi-armed bandit is a quintessential paradigm within the reinforcement learning literature. The goal of a bandit (which, for ease of exposition, we equate with the agent using the bandit algorithm) is typically to maximize its long-run utility in some multi-round setting given a set of actions $A$, while having no \textit{a priori} model of the environment in which it acts. Agents use a learning algorithm that balances exploring their action space---so as to learn the reward distribution associated with each action---with exploiting their current knowledge of these distributions. 

Bandit learning algorithms are typically implemented to minimize a measure of \textit{regret}, which in its standard definition is the long-run difference between the reward realized from an agent's actions and the reward from the actions that would have maximized expected utility. The process of action selection often follows \textit{action-value} methods, which are described formally below.

We model bandit learners playing the repeated Prisoner's Dilemma game while being unaware of the existence of the opponent, its choices, its outcomes, or any other aspect of the strategic nature of the interaction they are engaged in.  In each period $t$, an agent chooses an action $a_{t} \in A$.  After an agent takes an action, the environment supplies some reward value for this action. Specifically, the action-play vector $\alpha_{a}$ is a time-indexed indicator vector for a fixed action $a$ defined as $\alpha_{a,t} = \mathds{1}\{a_{t} = a\}$. For each player, there is one such vector for each action. The reward vector $\rho$ records the corresponding reward (or payoff) information, where $\rho_{t}$ is the reward at time $t$. For instance, $\alpha_H$ would be the indicator vector with components set to $1$ if $H$ was played in the round indexed by $t$, and $0$ otherwise. The dot-product of this action-play vector and reward vector leads to the cumulative sum of rewards for an action, which is useful for compactly computing value estimates.

Let us introduce the set $\mathcal{H}$, the \textit{history} for a bandit learner. The history stores self-action plays and rewards as
\[\mathcal{H} = \{\alpha_{a} | \; \forall \; a \in A\} \cup \{\rho\}.\]
Given a history, the bandit agent's estimate of the expected value of taking each action is called the value estimate $V$ and is defined as the empirical mean of rewards associated with playing the action (shown via the dot product of relevant vectors):
\[V_a(\mathcal{H}) = \begin{cases}\frac{\alpha_a\rho}{\alpha_a \vec{1}} & \text{ if } \alpha_a\vec{1} \neq 0 \\
0 & \text{otherwise,} \end{cases}\]
 where $\vec{1} = (1,1, ..., 1)$ of length $t$. 

Agents maintain and update two policies. In each time period, the agent estimates a \emph{target policy}, $\pi^*$, which specifies the action that has the highest average reward given the history.\footnote{We study \textit{context-free} agents in this model; for contextual bandits, this value is stored for each context state. In the case of multiple maximal-valued actions, various tie-breaking procedures may be used.} Formally, $\pi^*$ is defined for each timestep $t$ as
\[\pi^{*}_t(a) = \begin{cases} 1 &\text{if} \; a = \text{argmax}_{a \in A} \; V_a(\mathcal{H}_t)\\
0 & \text{otherwise}
\end{cases}.\]
For simplicity, we overload the notation of $\pi_t^*$ to denote the greedy \textit{action} at time $t$ (as opposed to the optimal distribution across actions). 

Theoretically, playing the target policy corresponds to strict exploitation of the information gathered by an agent so far. However, outside of ``commitment'' phases of certain algorithms following deliberate exploration, such a greedy policy is rarely used in practice because it foregoes any additional learning from deliberate exploration of the action space. Rather, to incorporate explicit learning by sampling rewards from the action space, agents follow a \emph{behavior policy} that defines how the agent behaves in an online setting. Different learning algorithms implement different behavior policies. 

We define such a bandit learning algorithm $\mathcal{A}$ as a mapping from a history $\mathcal{H}$ to a behavior policy $\pi \in \Delta(|A|)$, which is a probability distribution across actions. Such behavior policies typically balance exploration and exploitation, both to gather new reward information from the environment and to capitalize on current value estimates. 

We study the behavior of three broad types of behavior policies, each stemming from a certain type of bandit learning algorithm. Some algorithms devise a behavior policy which ensures that for any history, all actions have non-zero probability; we refer to agents using these algorithms as \textit{persistently random} bandits. Other algorithms may begin their play with random sampling but converge to selecting actions according to some degenerate distribution. We call agents which use this kind of algorithm \textit{greedy-in-the-limit}, consistent with standard nomenclature. 

\begin{definition}[Greedy-in-the-Limit]\label{def:GITL} A bandit agent is said to be \textit{greedy-in-the-limit} if it employs a learning algorithm $\mathcal{A}(\mathcal{H}_t)$ such that, under a fixed reward distribution, the policy converges to a limiting distribution $\mathcal{A}(\mathcal{H}_\infty)$ satisfying $\lim_{t \to \infty} \mathcal{A}(\mathcal{H}_t): = \pi_\infty^*$,  where $\pi_\infty^*$ is the observed greedy policy.
\end{definition}

Finally, some algorithms will map any history (at any point in time) to a degenerate distribution, and we refer to agents using such algorithms as \textit{deterministic}. Our analysis begins with \textit{persistently random} bandits, then turns to \textit{greedy-in-the-limit} bandits, and concludes with \textit{deterministic} bandits.

Standard bandit learning algorithms are largely \textit{mean-based}, where a behavior policy depends on average reward values per arm. These algorithms exhibit a more general property of \textit{path-invariance}, behaving in a manner that is invariant to the order in which a sequence of rewards was realized. 

\begin{definition}\label{def:path-invariant} A bandit agent $i$ is a \textit{path-invariant bandit} if learning algorithm $\mathcal{A}_i(\mathcal{H})$ produces the same policy output $\pi$ for all alternative orderings $\alpha^{'}_a, \rho^{'}$ of $\alpha_a, \rho \in \mathcal{H}$, where an alternative ordering $\alpha^{'}_a, \rho^{'}$ satisfies $\alpha_a \rho = \alpha_a^{'} \rho^{'}$ and $\alpha_a \vec{1} =  \alpha_a^{'} \vec{1} $.
\end{definition}
Two histories are \textit{path-equivalent} if they induce the same estimates in \textit{path-invariant} bandits. 

\begin{definition}\label{def:path-equivalent} Two histories $\mathcal{H}_1, \mathcal{H}_2$ are \textit{path-equivalent histories} if $\forall a \in A, \; \; \alpha_{1,a} \rho_{1} = \alpha_{2,a}\rho_{2} $ and $\alpha_{1,a} \vec{1} =  \alpha_{2,a} \vec{1} $. 
\end{definition}

\subsection{Learning Process}
We model the learning process of bandits playing a repeated Prisoner's Dilemma as a Markov Chain and analyze the dynamics of a walk on this chain. We index variables in the following order (when applicable): first by player ($i$), then by action ($a$), then by time ($t$). 

\paragraph{State Formulation}
For \textit{n} players, we create a state representation on the lattice $s_t \in \mathbb{N}^{2^{n}}$, which corresponds to the count of each outcome $o \in \{H,L\}^{n}$ at time $t$. This enumeration is illustrated for the two-player setting in Table \ref{tab:state_index}.  

\begin{table}[ht]
  \begin{tabular}{|c|c|c|}
    \hline
    \makecell{Element in $s_t$} & 
    \makecell{Occurrences of the outcome vector $o_t$ \\being counted by this state variable} & 
    \makecell{Reward vector $r_t$ associated with $o_t$} \\
    \hline
    $s_{t,1}$ & $(H,H)$ & $(\beta,\beta)$\\
        $s_{t,2}$ & $(H,L)$ & $(0,1)$\\
        $s_{t,3}$ & $(L,H)$ & $(1,0)$\\
        $s_{t,4}$ & $(L,L)$ & $(\gamma,\gamma)$\\
    \hline
    \end{tabular}
    \caption{Index of $s$ and corresponding outcome and reward vectors.}
    \label{tab:state_index}
\end{table}

For example, if agents have played $10$ rounds, with $4$ resulting in the outcome $(H,H)$, $3$ resulting in $(H,L)$, $2$ resulting in $(L,H)$, and $1$ resulting in $(L,L)$, then the state would be $s_{10} = (4,3,2,1)$. Importantly, this state construction is only visible to omniscient viewers of the game and not to the agents themselves, who are only able to observe their own actions and rewards.\footnote{For naive agents, full state reconstruction is impossible.}

\paragraph{Computing Value Estimates}
The state construction is analytically helpful because, owing to path-invariance, the state captures all information about history needed to compute value estimates, as summarized below.\footnote{These calculations of value estimates based on state rely on the denominators of the corresponding terms on the right-hand sides of the equations being non-zero.} 

\begin{table}[h!]
    \centering
    \setlength{\extrarowheight}{2pt}
    \begin{tabular}{c c}
    $V_{1,H}(s_t) = \frac{\beta s_{t,1}}{s_{t,1}+s_{t,2}}$ &  $V_{1,L}(s_t) = \frac{s_{t,3}+\gamma s_{t,4}}{s_{t,3}+s_{t,4}}$\\
    $V_{2,H}(s_t) = \frac{\beta s_{t,1}}{s_{t,1}+s_{t,3}}$ &  $V_{2,L}(s_t) = \frac{s_{t,2}+\gamma s_{t,4}}{s_{t,2}+s_{t,4}}$\\
    \end{tabular}
    \label{tab:value_function}
\end{table}
Our analytical results characterize long-run equilibrium outcomes. We first confirm that our construction of the play of the game as a Markov Chain is valid for all path-invariant bandits (which include all algorithms discussed in this paper).

\begin{lemma}\label{lem:markov}
With path-invariant bandits, $s_t$ adheres to the Markov property.
\end{lemma}

For a proof of this lemma, see Appendix \ref{prf:markov}. 

\subsection{Defining ``Collusion''}
With these value estimates, we define a \textit{collusive state} to be one where $V_{i,H}(s)>V_{i,L}(s) \; \forall \; i$ and a \textit{competitive state} to be one where $V_{i,H}(s)<V_{i,L}(s) \; \forall \; i$. However, defining these regions of the state space alone is insufficient to describe the long-run play that agents will settle on. Our motivating context of pricing algorithms leads us to focus on algorithms with some notion of rationality in their exploration over time. For this reason, we define \emph{learning to collude in the limit} as being characterized by the existence of some time $T$ after which $V_{i,H}(s)>V_{i,L}(s) \; \forall \; i$.\footnote{In what follows, we describe this limiting behavior interchangeably as the emergence of "seemingly collusive behavior," "collusive behavior" or "collusion." It should be clear from our model and results that none of these characterizations are implying the existence of anticompetitive collusion.} Put differently, for algorithms that learn to collude, there is some time $T$ after which $\pi_{1,t}^{*} = \pi_{2,t}^{*} = H$ for all $t>T$. This definition differs from other possible notions of \textit{apparent collusion}, which might, for example, measure how frequently $(H,H)$ is observed.

\subsection{Synchronicity}
Synchronicity is defined for each agent-action pair. It measures how often the other agent matches an agent’s chosen action in the same round, conditional on the agent taking the given action. For each agent-action pair, its number of rounds both played that action divided by the number of rounds the focal agent played that action. 

More formally, let $\xi$ be the transformation of the state space $S$ defined below. Denote $\xi_{i,H}$ to be the proportion of action samples in Player $i$'s history from action $H$ that were synchronous (i.e., the portion of $(H,H)$ plays over Player $i$'s total arm pulls of $H$). Let $\xi_{i,L}$ be the corresponding value for the $L$ action. Formally:

\begin{table}[h!]
    \centering
    \setlength{\extrarowheight}{2pt}
    \begin{tabular}{c c}
    $\xi_{1,H}(s_t) = \frac{s_{t,1}}{s_{t,1}+s_{t,2}}$ &  $\xi_{1,L}(s_t) = \frac{s_{t,4}}{s_{t,3}+s_{t,4}}$\\
    $\xi_{2,H}(s_t) = \frac{s_{t,1}}{s_{t,1}+s_{t,3}}$ &  $\xi_{2,L}(s_t) = \frac{s_{t,4}}{s_{t,2}+s_{t,4}}$\\
    \end{tabular}
    \label{tab:synchronicity}
\end{table}

This transformation provides an elegant characterization of a \textit{collusive region} in the transformed state space, ultimately in terms of this play-specific synchronicity.  The collusive region is illustrated in Figure \ref{fig:synchronous_space}. In particular, we now have a linear boundary on the region, with collusion defined as

\[
\beta \xi_{i,H} + (1-\gamma)\xi_{i,L} > 1, \quad \forall i \in \{1,2\}.
\]

\begin{figure}[ht]
    \centering
    \includegraphics[width=.5\textwidth]{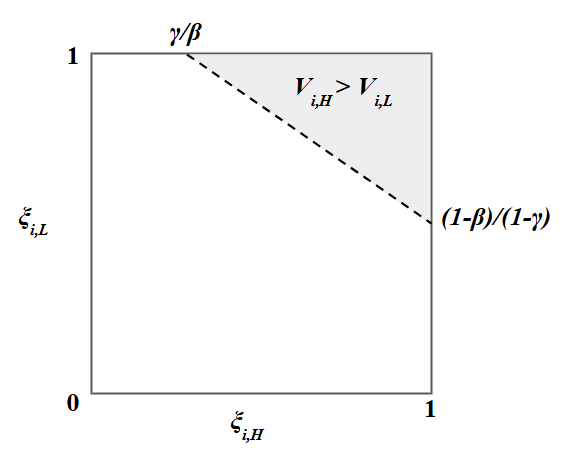}
    \caption{Region of $\xi_{i}$ where collusion is observed to produce a higher reward than competition.}
    \label{fig:synchronous_space}
\end{figure}

Our more involved definition of synchronicity is also motivated by the need to differentiate between collusive versus competitive outcomes, which precludes simply using a scalar count of same-action plays as our measure.\footnote{To see this, consider 
two states $s^1 = (10,1,1,0)$ and $s^2 = (5,1,1,5)$ in a game defined by $\beta = .5$ and $\gamma = .25$. Both states have the same number of outcomes in which agents chose the same action ($10$ states), but based on the definition of collusion above, $s^1$ is a competitive state while $s^2$ is collusive.}

\paragraph{No covariance guarantees no collusion}
We now show that there is a direct relationship between the empirical covariance in action plays as $t \rightarrow \infty$ and collusion as an emergent outcome. Specifically, we show that if realized action plays between agents are entirely uncorrelated, agents will never find collusion to be optimal, whatever the values of the game parameters $\beta$ and $\gamma$.

\begin{proposition}[No covariance guarantees no collusion]\label{prop:no-covariance}
    Let $\tilde{a}_{i,t} =\mathds{1}
\{a_{i,t}=H\} $ be the sequence of actions for Player $i$ up to time $t$, equal to $1$ if Player $i$ played $H$ in period $t$. Then
    $Cov(\tilde{a}_{1},\tilde{a}_{2})\leq0 \quad \Leftrightarrow\quad V_{i,H}(s_t)<V_{i,L}(s_t), \; \forall \;\beta, \gamma$.
\end{proposition}
The proof of this claim is in Appendix \ref{prf:no-covariance}. This general result provides conditions that ensure agents \textit{will not} collude for all Prisoner's Dilemma parameterizations and motivates the examination in the next section of a class of algorithms which experience no covariance between action plays in the limit.

\section{Persistently Random Algorithms}
Our synchronicity lens suggests a natural starting point: algorithms with more built-in randomness should be less likely to collude. We study the ``most random'' algorithms as a class with action selection shaped by persistent randomness. A \textit{persistently random} bandit employs a behavior policy that always places non-zero probability on all actions and, for a static reward distribution, will converge to a policy that places non-zero probability on each action. Formally, we define this class below.

\begin{definition}[Persistently Random Bandit]\label{def:persistently-random}
A bandit agent is said to be \textit{persistently random} if it employs a learning algorithm $\mathcal{A}$, mapping to a policy $\pi_{i,t}$, such that $\inf_{t,a}\pi_{i,t}(a)\geq \epsilon > 0$.  That is, agents retain a non-zero probability floor of selecting every action in any time period.
\end{definition}


An example of a persistently random algorithm is the \textit{epsilon-greedy} algorithm, which is commonly used as the basis for action selection in bandit and Q-learning algorithms. In devising a behavior policy, the epsilon-greedy approach incorporates randomness and is parameterized by some value $\epsilon \in (0,1)$. At time $t$, the agent selects the highest valued (greedy) action  with probability $1-\epsilon$; with probability $\epsilon$, the agent selects randomly and uniformly across all actions at its disposal, described below:
\begin{equation}
\label{eq:epsilon-greedy}
a_{i,t} = \begin{cases}\text{argmax}_{a \in A} \; V_{i,a}( \mathcal{H}_{i,t}) & \text{with prob.} \; \; 1-\epsilon \\
a, \;\forall \; \; a \in A & \text{with prob.} \; \frac{\epsilon}{|A|}
\end{cases}.
\end{equation}

Epsilon-greedy algorithms are a ``textbook'' persistently random algorithm, although other approaches could be conceived \citep{Sutton1998}. The evolution of value estimates for the epsilon greedy algorithm are shown in Figure \ref{fig:eg_play}.

\begin{figure}[ht]
    \centering
    \begin{subfigure}{0.49\textwidth}
        \includegraphics[width=\textwidth]{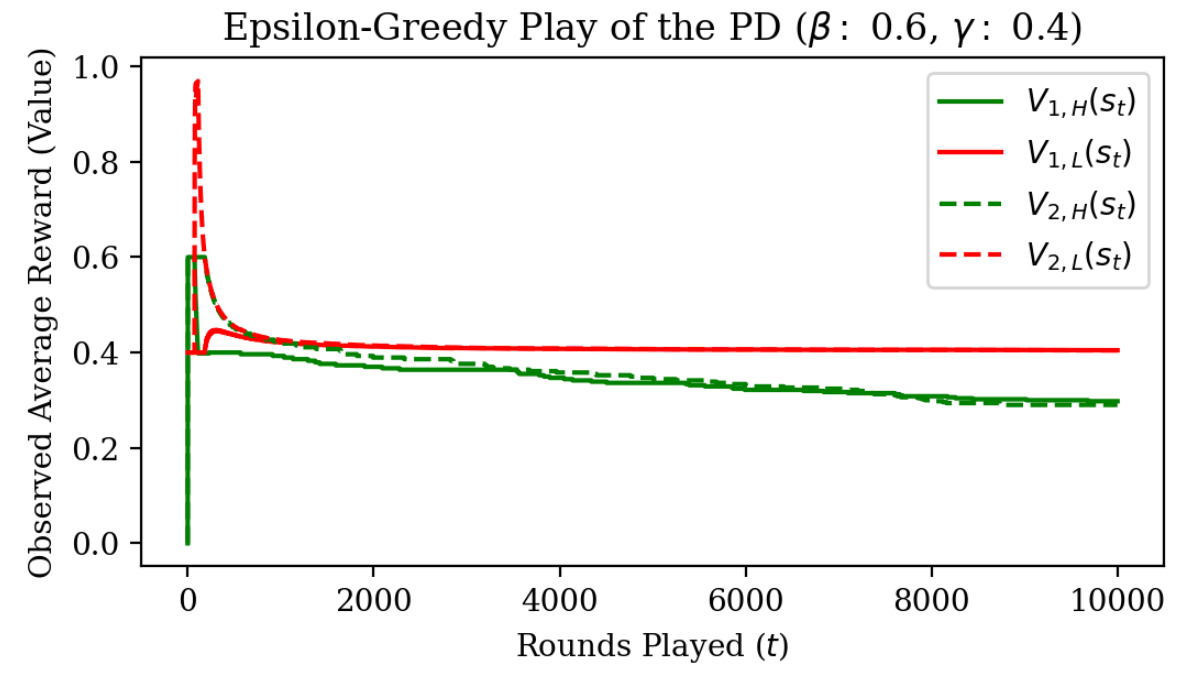}
    \end{subfigure}
    \hfill
    \begin{subfigure}{0.49\textwidth}
        \includegraphics[width=\textwidth]{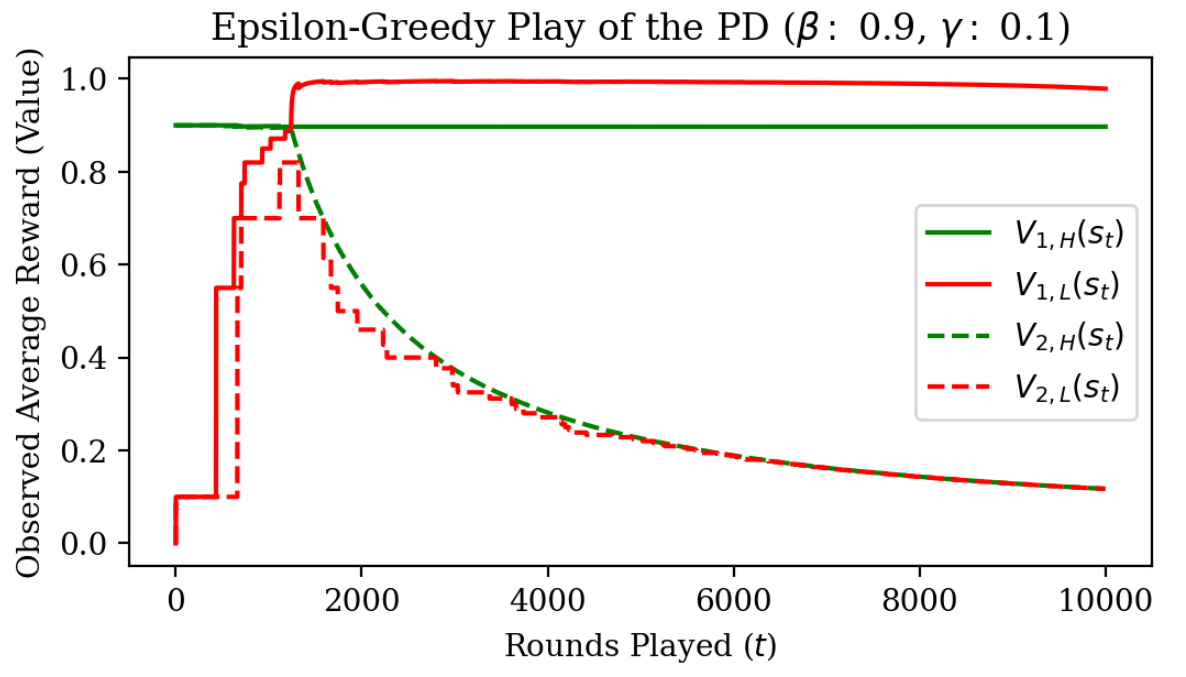}
    \end{subfigure}
    \caption{Value estimates for each agent's action across a 10,000 round Prisoner's Dilemma with epsilon-greedy agents ($\epsilon = .01$), for two sets of game parameters and epsilon values.  In both cases, the agents ultimately learn to compete (i.e., they learn that playing L has higher expected payoff).}
    \label{fig:eg_play}
\end{figure}

Our first key result establishes that the epsilon-greedy class of learning algorithms\footnote{Here, we refer to epsilon-greedy with constant epsilon; algorithms with a decaying epsilon value are discussed later under \textit{greedy-in-the-limit} analysis.} will always compete in the limit. 
\bigskip
\begin{proposition}[Epsilon-greedy will compete in the limit]\label{prop:epsilon-greedy}
  When epsilon-greedy bandits play an infinitely repeated Prisoner's Dilemma, as $t \rightarrow \infty$,  they will never learn to collude for any $\epsilon_1 > 0$, $\epsilon_2 > 0$.
\end{proposition}

The proof of the result analyzes each of the four regimes corresponding to greedy actions by agents, and follows from a higher expected value to playing $L$ than $H$ in any of these regimes, leading to a fixed joint distribution of $\Pr(a_i = H) = \tfrac{\epsilon_i}{2}$. The full proof is detailed in Appendix \ref{prf:epsilon-greedy}.  

Proposition \ref{prop:epsilon-greedy} is actually a special case of a broader setting in which both agents are employing\textit{ jointly convergent} persistently random algorithms, where agents each converge to playing a fixed random policy. We next show that the result generalizes to this broader setting. 

\begin{proposition}\label{prop:convergence-competition}
    Let $p_{i}$ be the sequence of $\{\pi_{i,t}(H)\}_{t=0}^\infty$. If agents are persistently random and  $(p_{1,t}, p_{2,t}) \xrightarrow[t\rightarrow\infty]{}(p_1^\star, p_2^\star)$ in the joint learning process, then $V_{i,L}(s_\infty)>V_{i,H}(s_\infty)\quad\forall \; \beta,\gamma, i$, and thus, agents will learn to compete in the limit.  
\end{proposition}

A proof of this proposition is in Appendix \ref{prf:convergence-competition}.

Our results might seem to suggest that in the presence of sufficient persistent randomness, naive algorithms will not learn to collude in the limit of their joint learning processes. There is reason to be cautious, however, since the results do not entirely preclude the emergence of seemingly collusive behavior. Without constraints on joint convergence, persistent randomness alone, even with rationality constraints such as imposing rank-ordered play, cannot guarantee convergence. An example of such an algorithm is provided in Appendix \ref{appendix:examples}, detailing how a designer could construct a non-persistently-random policy such that agents do not jointly converge to some set of limiting policies. In this case, these persistently random agents may still learn to collude. 

Moving beyond stylized counterexamples in which designers explicitly create collusive algorithms, a broader issue to consider is that algorithms used in practice often will \textit{not} be persistently random, and we may not want them to be. This is primarily due to the fact that in the standard bandit setting, 
a rational agent would not want to use an algorithm of this class because in this setting with a stationary reward distribution, the use of a persistently random algorithm cannot lead to sub-linear regret. In contrast, the online learning literature describes non-persistently random algorithms that are superior by guaranteeing sub-linear regret for the stationary setting. 

\section{Greedy-in-the-Limit Algorithms}
Non-persistently random algorithms still display randomness in their policies, but are asymptotically deterministic and converge to playing their greedy action. Guaranteeing sub-linear regret requires that the probability of playing actions that incur regret (i.e., sub-optimal actions) goes to zero in the limit. Randomized algorithms can achieve this by being \textit{greedy in the limit}, that is, by learning to play only the observed optimal action in the long run. 


\subsection{Explore-then-commit}
An example of an non-persistently random algorithm is the explore-then-commit (ETC) algorithm which functions similarly to an A/B test. For some prespecified number of periods, the algorithm samples actions non-adaptively, either uniformly at random (that is, each action is chosen with equal probability independent of history) or by cycling across actions. This sampling constitutes the \textit{explore} phase of the algorithm. Once the prespecified number of periods is reached, the algorithm \textit{commits} to the action with the highest average reward. 

ETC algorithms which deterministically cycle across actions during exploration (for instance, by using a modulo operator) fall into a more general class of \textit{deterministic} algorithms covered by results presented in Section 6. 
For the uniformly at random implementation, the formal action selection algorithm is
\begin{equation}
\label{eq:ETC-uniform-random}
  a_t = \begin{cases}
a, \;\forall \; \; a \in A & \text{with prob.} \; \frac{1}{|A|}, \;\; \text{if}\; t < t^*\\
\text{argmax}_{a \in A} \; V_a( \mathcal{H}_t) & \text{with prob.} \; \; 1, \;\; \text{if} \; t \geq t^*
\end{cases}.
\
\end{equation}

Figure~\ref{fig:etc-grid} illustrates the probability of collusion for a simulation of two symmetric ETC agents that play according to (\ref{eq:ETC-uniform-random}). Each tile corresponds to a set of game parameters and its color corresponds to the fraction of trials ($500$ per tile) in which agents learn to collude. Exploration is for $t^*$ rounds (both agents explore for the same number of periods), after which the agents commit to an action. The probability of collusion appears to decrease in exploration time, increase in $\beta$ and decrease in $\gamma$. 

\begin{figure}[htbp]
  \centering
  \begin{subfigure}[b]{0.45\textwidth}
    \includegraphics[width=\linewidth]{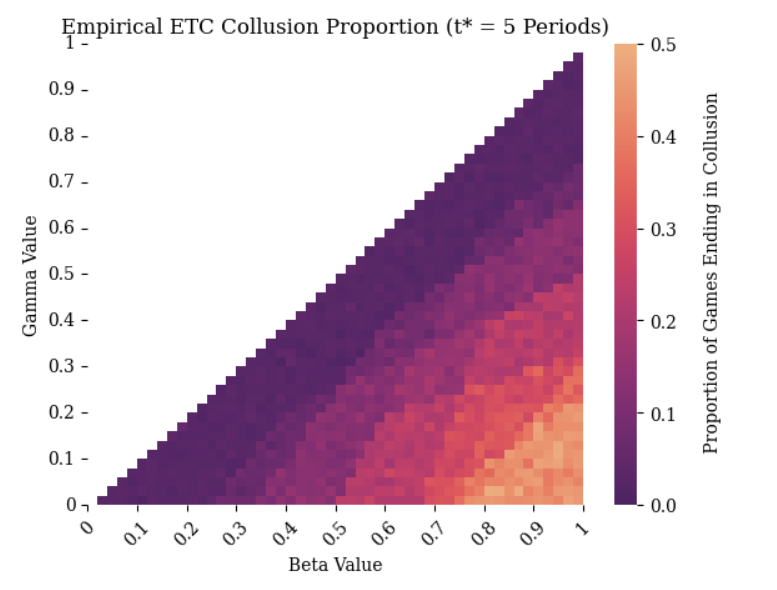}
    \caption{$t^*=5$}
  \end{subfigure}
  \hfill
  \begin{subfigure}[b]{0.45\textwidth}
    \includegraphics[width=\linewidth]{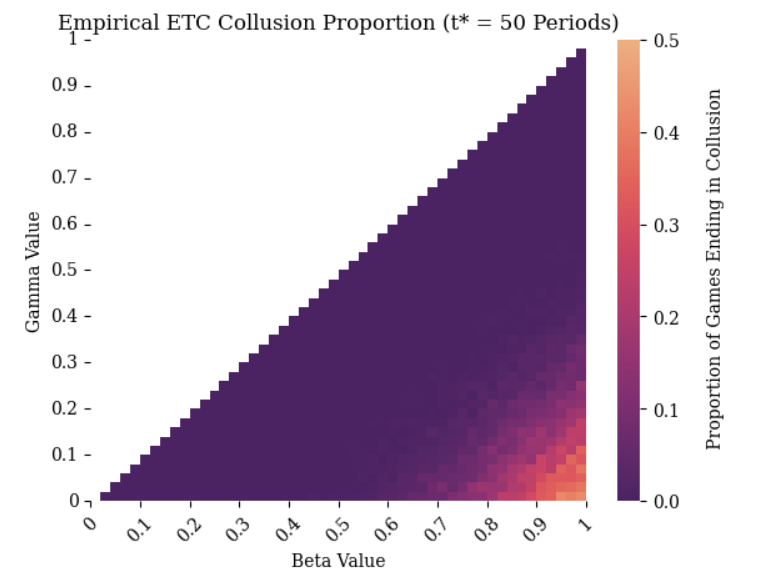}
    \caption{$t^* = 50$}
  \end{subfigure}
  
  \caption{Collusion probability by PD parameters for symmetric ETC agents}
  \label{fig:etc-grid} 
\end{figure}


To analyze the more general limiting behavior of symmetric agents that play according to (\ref{eq:ETC-uniform-random}), first note that during the exploration phase, there is a uniform probability across each outcome in $o$ for every exploration period. This is because sampling is uniformly random across actions for each agent with independent draws. Thus, the distribution across states at $t=t^*$ follows a multinomial distribution with a uniform probability on each outcome.  

Next, let $\bar{s}_t$ be the fraction of plays of each outcome at time $t$, obtained by dividing each component of the state $s_t$ by $t$. The likelihood of collusion is $\Pr(\bar{s}_{t^*} \in C|\beta, \gamma, t^*)$, where $C$ is the set of average states where both agents observe $H$ to be higher-valued, defined formally as

\begin{equation}
\label{eq:ETC-collusion-prob}
\begin{aligned}
C = \{\bar{s} \in \Delta_3, \frac{\bar{s}_1 \beta}{\bar{s}_1 +\bar{s}_2} > \frac{\bar{s}_3+\bar{s}_4\gamma}{\bar{s}_3+\bar{s}_4} \land  \frac{\bar{s}_1 \beta}{\bar{s}_1 + \bar{s}_3} > \frac{\bar{s}_2+\bar{s}_4\gamma}{\bar{s}_2+\bar{s}_4}\}.
\end{aligned}
\end{equation}
For $t^* = 1$,  $\Pr(\bar{s} \in C|\beta,\gamma,t^{*} = 1) = 0.25 \; \forall \; \beta, \gamma$. Additionally, consistent with Proposition \ref{prop:etc-exp}, as $t^* \rightarrow \infty$, there is no collusion.\footnote{As $t^* \rightarrow \infty, \bar{s}$ converges to $(.25,.25,.25,.25)$, which in turn yields value estimates $V_{1,H}(\bar{s}) = V_{2,H}(\bar{s}) = \frac{\beta}{2}$ and $V_{1,L}(\bar{s}) = V_{2,L}(\bar{s}) = \frac{1 + \gamma}{2}$.}  

For intermediate values of $t^*$, while a closed-form expression for the likelihood of collusion is not available, Figure \ref{fig:etc-grid} suggests that the probability is lower for higher values of $t^*$. This is proved in our next result. More precisely, we demonstrate that $\Pr(\bar{s}_{t^*} \in C|\beta, \gamma, t^*)$ is exponentially decreasing in $t^*$. 


\begin{proposition}\label{prop:etc-exp}
For symmetric ETC agents, 

(a) $\Pr(s_{t^*} \in C|\beta, \gamma, t^*) \leq c_1e^{-t^*c_2}$ for some constants $c_1$ and $c_2$ and

(b) An approximation of $\Pr(s_{t^*} \in C|\beta, \gamma, t^*)$ using a Gaussian approximation of the multinomial distribution takes the form 
\begin{equation}
\Pr(s_{t^*} \in C|\beta, \gamma, t^*)
   \;\approx\;
   \exp\!\{-t^*\,C(\beta,\gamma)\},
\end{equation}
where 
\[
C(\beta,\gamma)=
\frac{(1-\beta+\gamma)^{2}}
     {\,2\bigl[\beta^{2}+(1-\gamma)^{2}\bigr]
      +(1+\beta-\gamma)^{2}}.
\]
\end{proposition}
The proof of the result is presented in Appendix \ref{prf:etc-exp}. Intuitively, the result follows from the shrinking variance of this multinomial around the mean, which is $[.25,.25,.25,.25]$. This state is defined by no correlation in action plays, and so it never falls into the collusive region of the state space, regardless of game parameters, as per Proposition \ref{prop:no-covariance}. 

In a sense, the ETC algorithm and the epsilon-greedy algorithm are at two opposing ends of a spectrum of probabilistic exploration. There is no time variation in how epsilon-greedy allocates action plays to exploration, and thus, the design does not permit the algorithm to shift towards more productive exploitation as it learns over time from exploration. At the other end of the spectrum, the ETC algorithm concentrates exclusively on exploration early on, and then switches exclusively to exploitation after a point in time. One might consider whether there are algorithms that represent a middle ground, exploring more aggressively early on but gradually shifting to more active exploitation over time. One such algorithm, \textit{epsilon-greedy with decaying epsilon}, is analyzed in the following section.
\subsection{Epsilon-greedy With Decaying Epsilon}
A common variation of the standard epsilon-greedy algorithm involves an exploration probability that decreases in $t$, denoted $\epsilon(t)$. The standard approach achieves decay via geometric decay. The epsilon value is multiplied by a constant value $\eta \in (0,1)$ at every timestep. We make the standard assumption that $\epsilon(0) = 1$ for notational compactness and in order to parameterize this algorithm by a single parameter $\eta$. Thus, the exploration probability at time $t$ is denoted as

\[ \epsilon(t) = \eta^t .\]

At each time step, the action selection process follows a rule similar to that of the epsilon-greedy approach described in (\ref{eq:epsilon-greedy}), shown as

\begin{equation}
\label{eq:epsilon-decay}
a_t = \begin{cases}\text{argmax}_{a \in A} \; V_a(\mathcal{H}_t) & \text{with prob.} \; \; 1-\epsilon(t) \\
a, \;\forall \; \; a \in A & \text{with prob.} \; \frac{\epsilon(t)}{|A|}
\end{cases}.
\end{equation}

\begin{figure}[ht]
    \centering
    \begin{subfigure}{0.3\textwidth}
        \includegraphics[width=\textwidth]{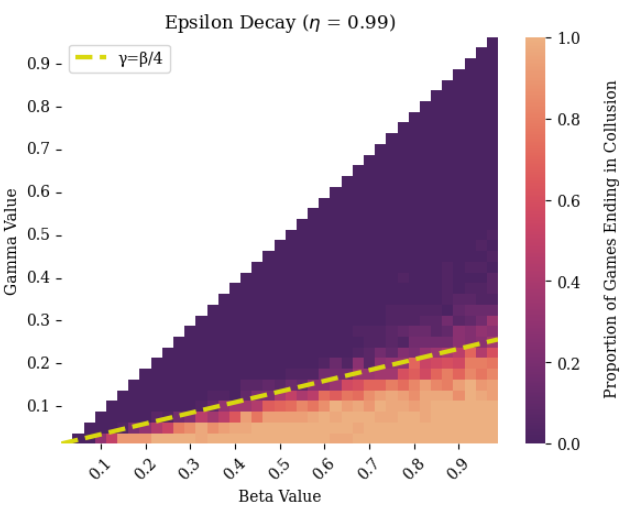}
    \end{subfigure}
    \hfill
    \begin{subfigure}{0.3\textwidth}
        \includegraphics[width=\textwidth]{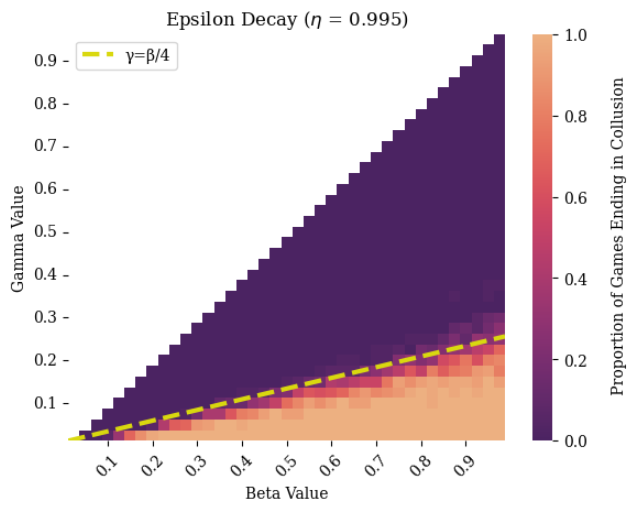}
    \end{subfigure}
    \hfill
    \begin{subfigure}{0.3\textwidth}
        \includegraphics[width=\textwidth]{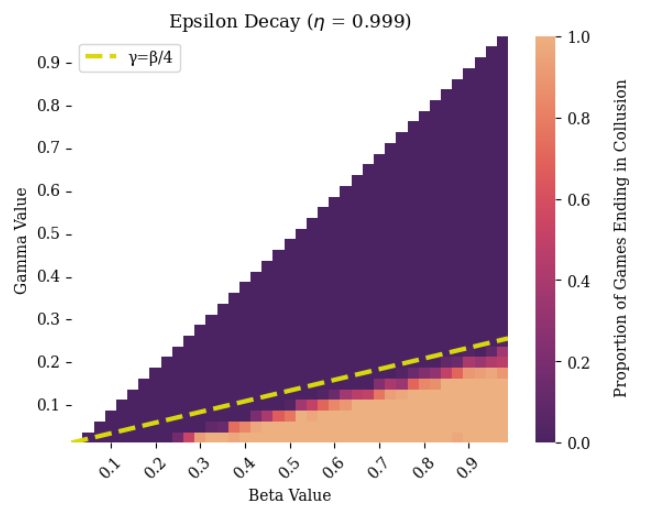}
    \end{subfigure}
    \caption{Proportion of games with epsilon-decay agents ending in a collusive equilibrium. Each plot demonstrates the empirical collusion probability across the game parameter space for a given decay rate of epsilon for both agents.}
    \label{fig:ed_play_grid}
\end{figure}

As described in Section 4, epsilon-greedy without decay always leads to non-collusive outcomes in the limit. However, when the exploration probability decays over time, the ensuing learning dynamics are subject to a more complicated structure. Over time, agents eventually settle on playing a single action as the probability of playing one of their exploratory actions approaches zero.\footnote{After 10,000 rounds, for $\eta = .999$, the probability of playing the non-greedy action is $2.2 * 10^{-5}$; for $\eta = .99$, this probability is $1.2 * 10^{-44}$.} However, both collusive and competitive outcomes are possible, as illustrated in Figure \ref{fig:ed_play_grid}, which report the empirical probability of settling on a collusive equilibrium for three different decay rates.  In this figure, each tile's color corresponds to the rate of collusion across $30$ trials, and each trial is comprised of $10,000$ rounds.

As the figure illustrates, as $\beta$ approaches $1$, the reward to colluding approaches that of playing $L$ while the opponent plays $H$, and agents are more likely to collude. Similarly, as $\gamma$ approaches $0$, the reward of competing approaches that of playing $H$ while the opponent plays $L$, and the likelihood of collusion increases. This relationship is consistent with intuition provided by Figure \ref{fig:synchronous_space} as an increase in the relative value of the cooperative outcome would increase the size of the collusive region of the state-space. As a result, the parameter space contains regions where collusion is very likely, regions where collusion is very unlikely, and boundary area where either result could plausibly emerge. This boundary follows approximately linear topographic lines and produces a more graduated transition from competition to collusion for lower $\eta$. A hard boundary can be approximated by the line $\gamma = \tfrac{\beta}{4}$, but we note that the $x$-intercept of this boundary increases for lower $\eta$. 

We can analyze the dynamics leading to this approximate boundary line. While, again, we have not derived a closed-form expression for the probability of collusion, we can explain the pattern uncovered in Figure \ref{fig:ed_play_grid} using an approximation. 
\begin{proposition}\label{prop:ed-beta}
The probability of collusion $\Pr(\bar{s}_\infty \in C)$ can be approximated as: 

\[\Pr(\xi_{1,H}(s_\infty)>\tfrac{\gamma}{\beta})
\;\approx\;
1-\mathrm{B}\!\bigl(\tfrac{\gamma}{\beta};\;\lambda_1,\lambda_2\bigr),\]

where $\lambda_1=\frac{1}{4(1-\eta^{2})}, \lambda_2=\frac{1+2\eta}{4(1-\eta^2)}$, and
$\mathrm{B}(x; \cdot)$ denotes the cumulative distribution
function of a Beta random variable with the
parameters $\lambda_1,\lambda_2$.
\end{proposition}
Under this approximation, $\mathbb{E}[\xi_{i,H}] = \frac{\lambda_1}{\lambda_1 + \lambda_2} = \frac{1}{2(1+\eta)}$ with small variance, as given by the Beta approximation.

A neat heuristic emerges from the proposition above. If we let $\eta \rightarrow 1$ and assume that variance synchronicity on the $H$ arm shrinks to $0$, then we have collusion likely emerging in the region $\frac{\gamma}{\beta} < \frac{1}{4}$. This describes our visual evidence from simulations presented in Figure \ref{fig:ed_play_grid} well, with a wider distribution in lower $\eta$ contributing to this ``softer'' boundary. More broadly, this analysis confirms the graphical depiction of collusion being the likely outcome when $\beta$ is high relative to $\gamma$. However, this approximation also describes a challenge in simplistic monitoring strategies for algorithmic collusion. Note that the limiting collusive behavior is actually \textit{induced by early competition increasing synchronicity on the $L$ arm.} In other words, early competitive behavior under certain payoff conditions actually \textit{leads} to collusion in the long run.  

While increased random exploration decreases collusion potential for ETC algorithms, with collusion probability decreasing in $t^*$, the same dynamic does not hold for epsilon-decay algorithms, as Figure \ref{fig:ed_play_grid} and Proposition \ref{prop:ed-beta} demonstrate. Thus, while persistent randomness in action selection guarantees competition in the limit, this no longer holds for greedy-in-the-limit algorithms which remove randomness from action selection in the limit. The next section extends this analysis further, examining algorithms that exhibit no randomness in action selection at all.

\section{Deterministic Algorithms}\label{sec:deterministic}
Many standard bandit algorithms follow a strictly deterministic approach for action selection. When a behavior policy follows a degenerate distribution (that is, at each step the algorithm prescribes exactly one action with probability $1$) for any realized history, we call the agent employing it a \textit{deterministic} bandit. Even without randomness, these algorithms may still explicitly explore the action space, such as algorithms that greedily select actions from index functions that incorporate exploration bonuses based on sample size. An example of these are \textit{upper confidence bound} algorithms, which are widely used for their ``smart'' experimentation methods and are deterministic in their history. 

\begin{definition}[Deterministic Bandit]\label{def:deterministic} A bandit agent is said to be \textit{deterministic} if it employs some learning algorithm $\mathcal{A}$ where $\mathcal{A}(\mathcal{H})$ follows a degenerate distribution for any $\mathcal{H}$.
\end{definition}

To keep our analysis aligned with the exploration-exploitation trade-off that is central to any sensible bandit algorithm, we restrict our attention to algorithms that will play each action at least once.  The \textit{upper confidence bound (UCB)} action selection process is as follows~\citep{lattimore2020bandit}.  Define:

\[\text{UCB}_a( \mathcal{H}_t) = \begin{cases} V_a(\mathcal{H}_t) + \sqrt{\frac{2\log 1/\delta}{\alpha_a\vec{1}}} & \text{if} \; \; \alpha_a \vec{1} \neq 0\\
 +\infty & \text{otherwise} \end{cases}.\]

Here, $\delta \in (0,1]$ serves as an exploration parameter that defines the width of the confidence interval. A lower $\delta$ lends more weight to sampling under-sampled actions. UCB selects actions based on:

\[a_t = \text{argmax}_{a \in A} \text{UCB}(a, \mathcal{H}_t),\]

with the \textit{argmax} function returning the maximal UCB-valued action, or some selection of the maximal UCB-valued actions when there are multiple actions with the same maximal UCB value. A common variant of this confidence bound construction, UCB1, is described in Appendix \ref{appendix:examples}.

\subsection{Symmetric deterministic algorithms}
We begin with a result that is the converse of Proposition \ref{prop:convergence-competition}, where we saw that jointly convergent  persistently random bandits will learn to compete. Here,  for any set of payoffs, symmetric, deterministic bandits (bandits that employ the same deterministic algorithm) will always learn to collude. To illustrate this process, we begin with a series of lemmas. 
\begin{lemma}\label{lem:monotonicity-H}
    For each player $i$, when $o_t = (H,H)$, either: \newline 1.  $V_{i,H}(s_{t+1}) > V_{i,H}(s_{t})$  or \newline 2. $V_{i,H}(s_{t+1}) = V_{i,H}(s_{t}) = \beta$.
\end{lemma}

For a proof of this claim, see Appendix \ref{prf:monotonicity-H}. We observe a parallel form of monotonicity for the value estimates of $L$ which is proved in Appendix \ref{prf:monotonicity-L}.

\begin{lemma}\label{lem:monotonicity-L}
    For each player $i$, when $o_t = (L,L)$, either: \newline 1. $V_{i,L}(s_{t+1}) < V_{i,L}(s_{t})$  or \newline 2. $V_{i,L}(s_{ t+1}) = V_{i,L}( s_{t}) = \gamma$ .
\end{lemma}  

We next establish that symmetric, deterministic bandits will play the same action in all subsequent periods if they have observed path-equivalent histories. 
\begin{lemma}\label{lem:subsequent}If at some period $t^{*}$, symmetric, deterministic, and path-invariant bandits have path-equivalent histories, they will play the same action in every subsequent period $t\geq t^{*}$.
\end{lemma}

These lemmas lead to our general result, that deterministic bandits learn to collude:

\begin{proposition}\label{prop:deterministic}When symmetric and deterministic bandit players play an infinitely repeated Prisoner's Dilemma, there exists some finite time period $T$ such that for any $t \geq T$, $\pi_{1,t}^{*} = \pi_{2,t}^{*} = H$. 
\end{proposition}

Next, consider interactions between bandits employing the UCB learning algorithm. While these learning algorithms are typically deterministic, tie-breaking rules are used when the upper confidence bounds of two or more actions are both maximal and equal, as is the case when an agent's history is empty. The tie-breaking rule could be symmetric and deterministic, and this would be covered by the general proof in Proposition \ref{prop:deterministic}. However, these tie-breaking rules could also be asymmetric or could introduce an element of randomization. In our specific setting, the tie-breaking rules for each agent could be identical, or they could be different. Both randomness and asymmetry in tie breaking create more complex learning dynamics, notably at the start of the game.  There is the possibility for initial competition amongst agents; however, we show they will ultimately settle on collusion even if tie breaking is random or otherwise asymmetric. We now prove that the agents will always converge to the collusive outcome, conditioned on their choice of a reasonable $\delta$.

\begin{proposition}\label{prop:UCB}
    When symmetric UCB bandits play an infinitely repeated Prisoner's Dilemma, there exists some finite time $T$ such that for all $t \geq T$, $\pi_{1,t}^{*} = \pi_{2,t}^{*} = H$, for $\delta<e^{\frac{-\gamma^2}{2}}$.
\end{proposition}
We arrive at this result via a case-by-case analysis of the start of the game. We show that for $t \geq 2$, players will play the same action, so we invoke a similar argument to that used in Proposition \ref{prop:deterministic}. See Appendix \ref{prf:UCB} for a detailed proof.

\subsection{The effects of asymmetry}
Our framework provides a foundation for understanding how collusion emerges as the learned optimal action by naive bandit learners. However, certain algorithmic settings, such as asymmetric parameterization or offset starts of UCB agents do not fall neatly within the theoretical framework. Namely, these behaviors appear to be highly non-linear and possibly chaotic. Therefore we provide simulated results of these behaviors below.  In doing so, we illustrate that while the certainty of the result of Proposition \ref{prop:deterministic} relies on lock-step-action assumptions, deterministic algorithms can very often lead to learned collusion even when these assumptions do not hold.
\paragraph{Asymmetric UCB Parameters}
Extending the analysis to asymmetric UCB algorithms, let's first examine what happens when the algorithms may have different values for the exploration bonus parameter $\delta$. Results of a simulation are presented in Figure \ref{fig:asymmetric_ucb}. Parameters are chosen uniformly at random in the region $\beta \in (0,1), \gamma \in (0,\beta), \delta_1 \in (0,1), \delta_2 \in (0,\delta_1)$ with $T=10,000$. Player 1 will denote the player with the greater exploration bonus, so $\delta_1>\delta_2$ to maximize the efficiency of simulations (otherwise, there would be repetition due to symmetry). 

Figure \ref{fig:asymmetric_ucb} shows the results of this simulation.  In these runs, $41.3\%$ of games ended in collusion, across $73,000$ trials.  The results show that collusion potential varies across parameters, with no monotone guarantees in $\beta$ or $\gamma$ and few immediately identifiable patterns in how asymmetry in $(\delta_1,\delta_2)$ affects collusion potential. Here, tiles along the line $\delta_1 = \delta_2$ do not correspond to perfectly symmetric parameterization, as these tiles bin the uniformly random draws of $(\delta_1,\delta_2)$. However, the empirical results here do suggest that near symmetry influences collusion probability for certain regions of the game parameter space.  Direct empirical validation of the symmetry claim in Proposition \ref{prop:UCB}  can be found in Figure \ref{fig:asymmetric_ucb_fixed} in Appendix \ref{appendix:examples}, where $\delta$ values are selected along the grid and game parameters are fixed. 

It is worth noting that because these algorithms are deterministic, for a given $(\beta, \gamma, \delta_1, \delta_2)$, there is a deterministic outcome, collusion or competition. Tiles that show a collusion proportion in $(0,1)$ indicate that different payoff parameters within the tile result in different outcomes. These intermediately shaded tiles indicate that collusion potential is \textit{highly} sensitive to setting parameters. The graphs in Figure \ref{fig:asymmetric_ucb} demonstrate this sensitivity to game parameters; however, Figure \ref{fig:asymmetric_ucb_fixed} in Appendix \ref{appendix:examples} demonstrates this sensitivity to $(\delta_1,\delta_2)$ as well.

\begin{figure}[ht]
    \centering
    \begin{subfigure}{0.3\textwidth}
        \includegraphics[width=\textwidth]{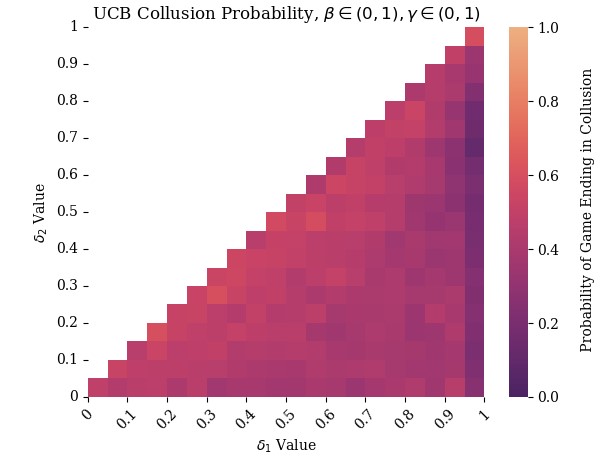}
        \subcaption{With $\beta \in (0,1), \gamma \in (0,1)$, $41.32\%$ of games end in collusion ($n = 73,000$)}
    \end{subfigure}
    \hfill
    \begin{subfigure}{0.3\textwidth}
        \includegraphics[width=\textwidth]{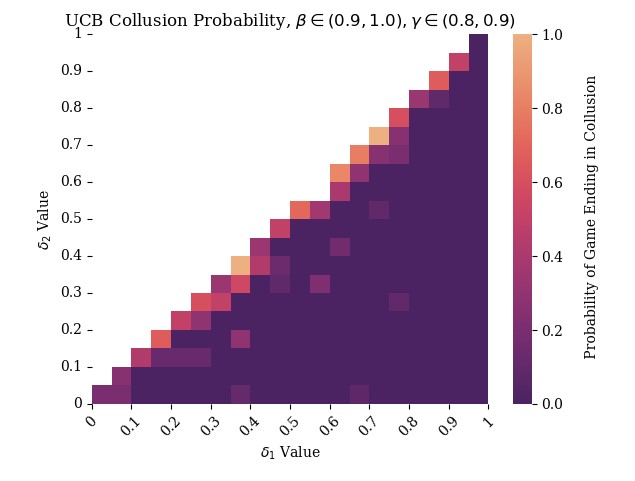}
        \subcaption{With $\beta \in [.9,1), \gamma \in [.8,.9)$, $6.12\%$ of games end in collusion ($n = 1,455$)}
    \end{subfigure}
    \hfill
    \begin{subfigure}{0.3\textwidth}
        \includegraphics[width=\textwidth]{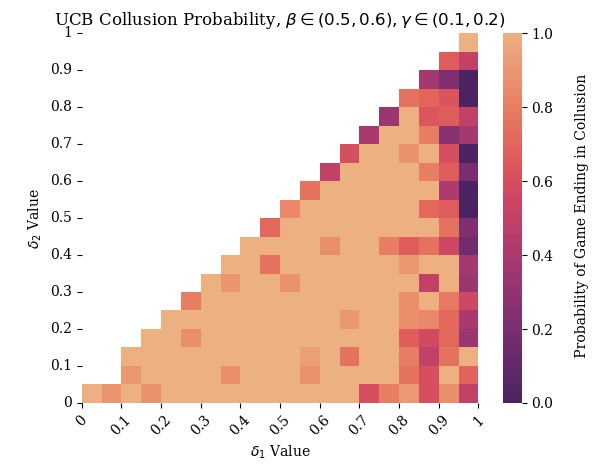}
        \subcaption{With $\beta \in [.5,.6), \gamma \in [.1,.2)$, $85.2\%$ of games end in collusion ($n = 1,443$)}
    \end{subfigure}
    \caption{Proportion of games with asymmetric UCB agents ending in a collusive equilibrium across $(\delta_1,\delta_2)$. Each graph displays a range of game parameters. The left graph displays the full range of game parameters; the center displays games with the highest range of $\beta$ and a high range of $\gamma$ values; the right displays games with a middle range of $\beta$ and the lowest range of $\gamma$.}
    \label{fig:asymmetric_ucb}
\end{figure}

\paragraph{Offset UCB Start}
With offset start periods, UCB can still collude with high probability in many regions. However, in the offset-start setting, the iterated Prisoner's Dilemma is not well defined. Specifically, there is no mapping defined from one player's actions to their reward. To address this, we construct a reasonable estimate of what the Prisoner's Dilemma reward \textit{would} be in the setting of pricing. When both agents play the same action, there is an implicit assumption that they each capture half the market willing to buy at that price. As a result, we assume that the whole market at that price purchases from the sole firm. So, we model the reward for each action in the single agent setting as double the symmetric payoff in the Prisoner's Dilemma. Concretely, we construct the mapping $f: A \rightarrow \mathbb{R}$ to define the single agent payoffs as

\[
f(a) = \begin{cases} 2\beta &\text{if} \quad a = H,\\
2\gamma &\text{if} \quad a = L
\end{cases}.
\]

To simulate, we run $71,000$ trials drawn from a symmetric $\delta \sim \text{Uniform}(0,1)$ and offset drawn uniformly from the integers $\{1,...,99\}$.  This offset corresponds to the number of rounds where Player 1 is acting and learning in the environment before Player 2 joins and the game reverts to the Prisoner's Dilemma. We again use $10,000$ rounds per trial (for Player 1), inclusive of the offset, and draw payoff parameters according to Figure \ref{fig:asymmetric_ucb}. Across these trials, we have $37.7\%$ of games ending in collusion. In Figure \ref{fig:offset_ucb}, these trials are decomposed into two graphs, one with a shortest quintile of offsets, and the other with the longest. Both graphs present a discernible, but highly non-linear pattern of collusion counts in game parameters. Surprisingly, this pattern seems relatively insensitive to offset values.

\begin{figure}[ht]
    \centering
    \begin{subfigure}{0.45\textwidth}
        \includegraphics[width=\textwidth]{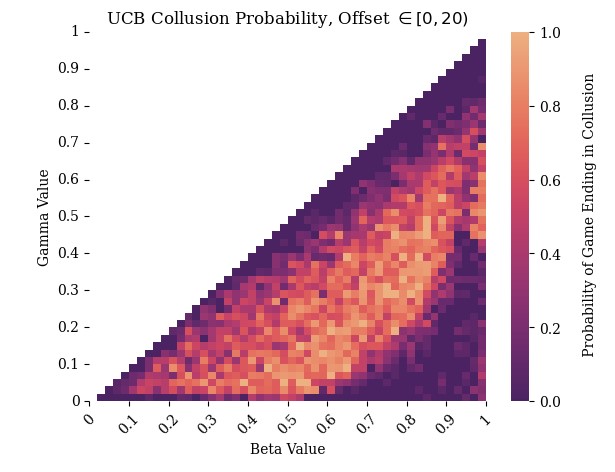}
        \subcaption{With $\text{offset} \in [0,20) $, $38.61\%$ of games end in collusion ($n = 13,599$)}
    \end{subfigure}
    \hfill
    \begin{subfigure}{0.45\textwidth}
        \includegraphics[width=\textwidth]{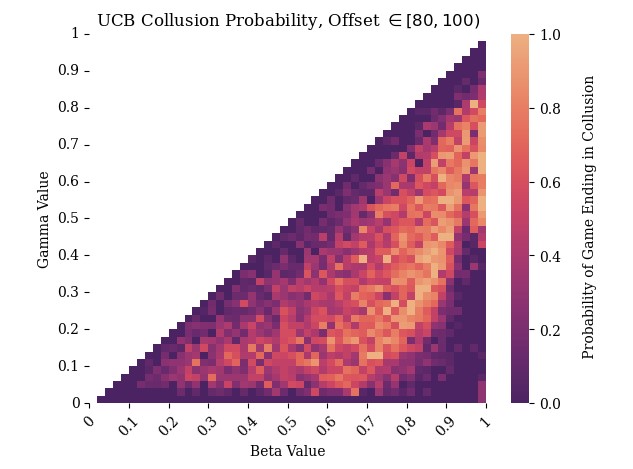}
        \subcaption{With $\text{offset} \in [80,100) $, $35.99\%$ of games end in collusion ($n = 14,433$)}
    \end{subfigure}
    \caption{Proportion of games with offset-start UCB agents ending in a collusive equilibrium across $(\beta,\gamma)$. The left graph denotes a small offset of periods where only Player 1 is interacting, while the right graph denotes a large offset.}
    \label{fig:offset_ucb}
\end{figure}

\section{Discussion}
This paper shows that even absent any information whatsoever about the strategic interaction in which they are engaged, competing algorithms may nevertheless converge to seemingly collusive outcomes. However, whether or not seemingly collusive outcomes will occur, or are likely, is closely related to the level of randomness in the algorithm's learning policy. Naive algorithmic collusion \textit{always} emerges in the long-run when competing symmetric agents both employ a deterministic bandit learning algorithm, sometimes emerges when the learning process of the algorithm is non-persistently random, and \textit{never} emerges in the long-run when both agents use a broad class of persistently random algorithms. The general connection between randomness and the emergence of seemingly collusive outcomes is aligned with prior theory~\citep{green1984noncooperative} describing how noise can undermine collusive behavior, although in our models, noise is non-stationary. The exact dynamics of the joint learning process affect how the noise distribution evolves and ultimately whether collusion is reached. We also highlight that seemingly collusive outcomes can be path-dependent. In the limit, exactly the same pair of algorithms may generate seemingly collusive outcomes in certain instances and competitive outcomes in others. 

Our results also show that \textit{synchronicity} in choices while agents are learning action values---a new outcome metric we highlight---plays an important role in determining whether agents learn to compete or to collude. We highlight that the choice of experimentation algorithm is central to the ensuing synchronicity in action plays between agents. Persistently random algorithms that jointly converge to playing static policies will always learn to compete. But, as algorithms become greedy in the limit, a requirement for random algorithms to be regret-minimizing, the result changes. For common greedy-in-the-limit algorithms, collusion emerges as a potential, and often likely, learned outcome.  At the other end of the spectrum, symmetric deterministic algorithms will always learn to collude, and the introduction of a small amount of randomness (for instance, the randomness induced by UCB tie breaking) is not sufficient to prevent collusion in the limit. 
\subsection{Extensions}
As a widely used game that captures a tradeoff between cooperation and competition, the Prisoner’s Dilemma serves as an effective model for our analysis of the emergence of collusive behavior. We believe our findings would extend to more particular models of strategic interaction that feature a comparable tradeoff. 

To illustrate, consider competing bandit learners in a setting for which the reward to agent $i$ in period $t$ is specified by a standard logit model that one might encounter in the industrial organization literature \citep{calvano2020artificial}:
\[
r_{i,t} = \left(\frac{e^{-\alpha a_{i,t}}}{\sum_{j = 1}^{n}e^{-\alpha a_{j,t}}} + \mathcal{N}(\vec{0},\sigma^2)\right)(a_{i,t}-c_i),
\]
where $\alpha$ is a demand sensitivity parameter and $c_i$ is the marginal cost of selling a good to agent $i$.

 As an example, we fix $n=2$ and $\alpha = 1$, $c_i = 1\;  \forall \; i$, and examine changes in outcomes as the variance of noise $\sigma$ varies. Players are given a lattice in $[0,10]$ with $100$ equally spaced prices as the action space $A$ and use symmetric algorithms.  This parametrization yields a Nash equilibrium of $(3,3)$. The simulations are run for $100,000$ periods with $100$ trials for each algorithm type. 

Figure \ref{fig:logit-outcomes} illustrates the learning behavior of various algorithms under this reward model and parametrization. We see the intuitions of our analytical models verified in the simulation. Epsilon-greedy variants tend to learn to charge near Nash prices, while epsilon-decay algorithms will, in general, learn to price higher. These prices tend to be slightly higher for $\eta$ closer to $1$, which confirms the intuition from our analytical models. ETC algorithms tend to learn to price higher given a shorter experimental window, again confirming the intuition of our analytical results. Finally, in the no-noise environment, UCB learns to charge the maximum price, aligned with our result that symmetric play by deterministic algorithms will lead to collusion.  However, when noise is added, UCB appears generally to price lower than ETC algorithms.  More generally, we see that noise in demand scatters the resulting learned optimal prices. These simulation results show that naive collusion can also emerge in more sophisticated models of interaction and that, again, the outcome depends starkly on the learning algorithms being employed. 
\begin{figure}[htbp]
  \centering
  \begin{subfigure}[b]{0.45\textwidth}
    \includegraphics[width=\linewidth]{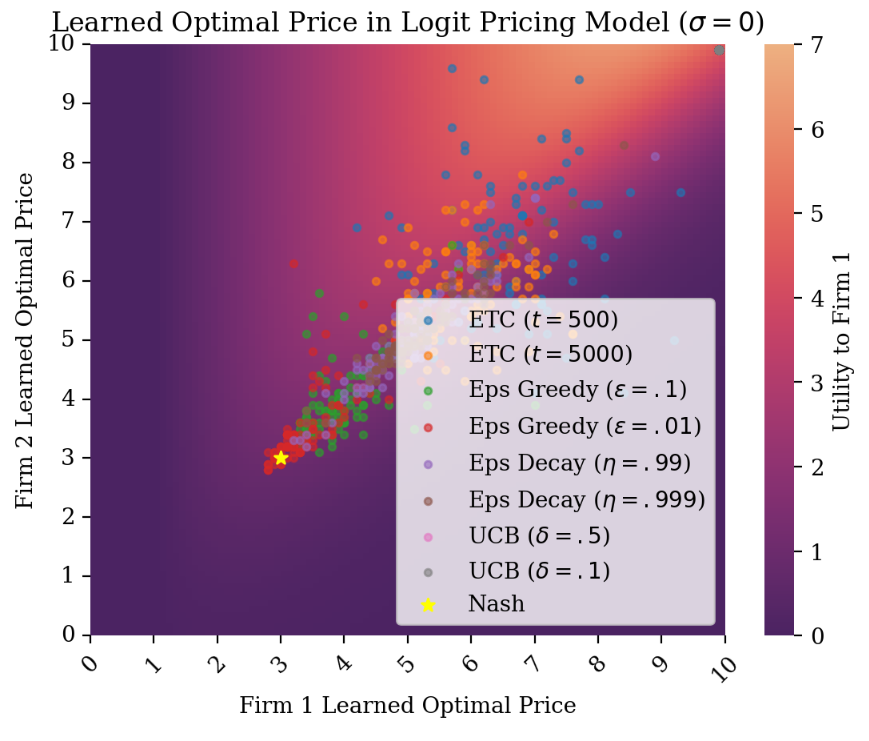}
    \caption{$\sigma = 0$}
  \end{subfigure}
  \hfill
  \begin{subfigure}[b]{0.45\textwidth}
    \includegraphics[width=\linewidth]{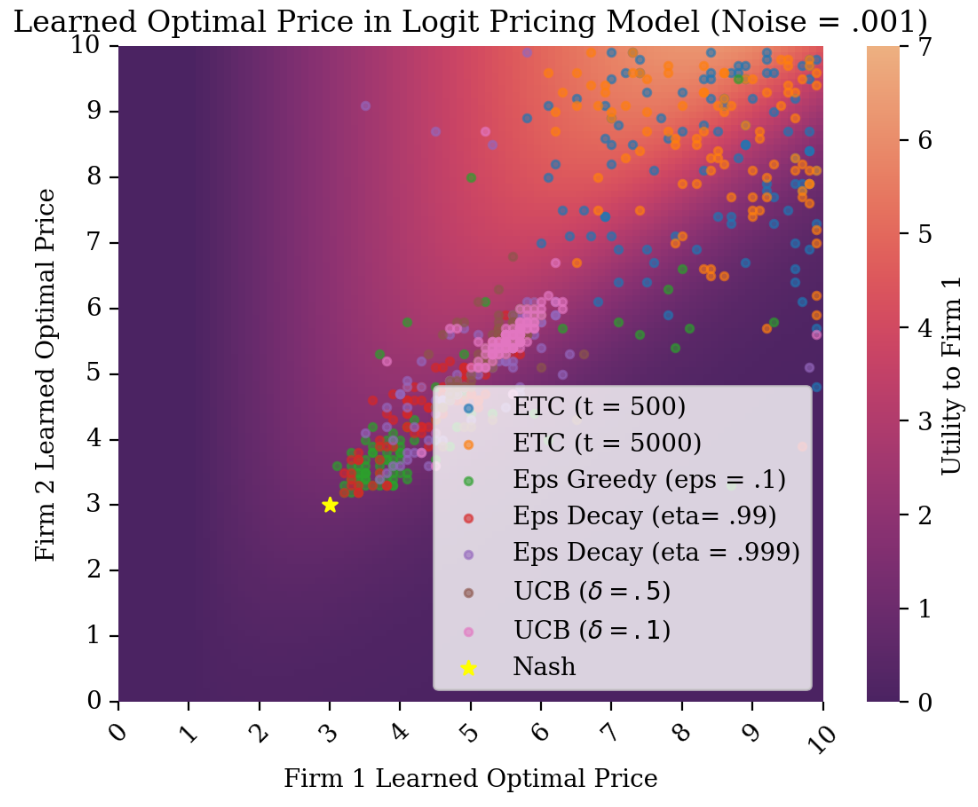}
    \caption{$\sigma = 0.001$}
  \end{subfigure}
  \caption{Points depict the learned-optimal action by both agents in logit pricing model after $100,000$ periods. The color underlay represents the expected payoff to firm $1$ across $(a_1,a_2)$, where payoffs are symmetric across firms. The choice of experimentation algorithm affects the learned-optimal payoff, both with and without noise.}
  \label{fig:logit-outcomes}
\end{figure}
\subsection{Policy implications}
Our work has several key implications for understanding and regulating algorithmic collusion in online marketplaces.

\paragraph{Algorithmic collusion can be ``naive."} We provide a thorough treatment of naive collusion, studying the dynamics of standard algorithms for addressing decision making under uncertainty (bandit learning agents). In this treatment, we highlight that firms running independent pricing experiments with various ``textbook'' experimentation algorithms can often learn to price in a manner that yields seemingly collusive outcomes. We highlight that work done by \citet{hansen21} is not simply a knife's edge result; this phenomenon can, with many algorithms and across many market settings, lead to seemingly collusive outcomes. Consequently, limiting algorithms from conditioning their choices on opponent prices or interventions that penalize algorithms that punish opponent deviations will be insufficient to completely prevent collusive outcomes from arising. Rather, our findings indicate that a more nuanced understanding of market and algorithmic settings is required to understand the potential for tacit algorithmic collusion.

\paragraph{Trial-and-error exploration can have countervailing effects.} Prior work on algorithmic collusion has sometimes focused on the \textit{increased} risk of observing collusive outcomes when there is a higher level of ``trial-and-error" exploration, owing to the algorithms potentially being given greater latitude to discover collusive outcomes and tacit coordination strategies.\footnote{For example, \citet{calvano2020artificial} note that ``No explicit communication or agreement is required: algorithms learn by trial and error how to raise prices and avoid undercutting.'' Similarly, Klein (2021) suggests that ``Q-learning agents engage in random exploration, which allows them to experiment with price increases that would not be chosen under myopic best-response behavior,'' and that through such exploration, ``firms learn that higher prices can be profitably sustained if rivals respond in kind.''} In contrast, our results show that more extensive trial-and-error exploration can sometimes have an opposing effect on the emergence of seemingly collusive outcomes. In particular, as we have illustrated in Figure \ref{fig:etc-grid}, there are settings under which a longer period of ``trial-and-error" reduces, rather than increases, the likelihood of observing seemingly collusive outcomes. At the same time, we also find that when using the epsilon-greedy algorithm with decay, a higher $\eta$ (mapping to more cumulative exploration) can actually bias synchronicity in a manner that brings agents closer to the threshold where they observe colluding to be better than competing. These countervailing findings, while also challenging those in earlier papers like \citet{banchio2023a}, suggest that policymakers must carefully consider the varied nuances in the relationship between exploration and seemingly collusive outcomes. 

\paragraph{Symmetry leads to seemingly collusive outcomes.} We show that increasing synchronicity in action plays as agents learn increases the likelihood of seemingly collusive outcomes. Symmetric, deterministic algorithms will always collude in our model, while agents that are persistently random will not learn to collude in the limit. This finding indicates that observed seemingly collusive outcomes do not necessarily imply that pricing algorithm distributors have imbued their algorithms with supracompetitive tendencies. Rather, simply by using algorithms that will tend to price similarly in the same situations, firms in the same market using a learning algorithm provided by a central distributor may generate outcomes that feature supracompetitive prices.

\paragraph{Algorithm choice affects the emergence of seemingly collusive outcomes.} This discussion of symmetry leads more generally to a discussion of algorithm choice. Our findings indicate that the exact choice and parameterization of learning algorithms shapes the emergence of seemingly collusive outcomes. With some algorithms and in some market settings, one can conclude that seemingly collusive outcomes definitively will or will not arise even from independent learning processes. In other settings, however, this outcome is non-deterministic: it cannot be said \textit{ex ante} whether the learning process will lead to collusive outcomes. Our findings may therefore provide guidance about what market and algorithmic conditions might benefit from greater regulatory scrutiny.



\subsection{Future work}

\paragraph{Increasing number and heterogeneity of agents.} Do the present results generalize to many competing agents? How does the size and composition of a pool of homogeneous or heterogeneous learning agents affect outcomes? We conjecture that an increasing number and heterogeneity of agents would decrease synchronicity of plays and thus decrease collusion potential, but it also may depend on other details of the game.  What if agents had access to customer characteristics?  Instead of colluding on high prices, they may end up splitting the market in a (seemingly) collusive manner. In the presence of such agents, it would also be interesting to see how non-algorithmic agents affect learned outcomes, similar to what is studied by \citet{wang_algorithms_2023}. Additionally, it may be insightful to study additional sampling algorithms, such as Thompson sampling. We excluded Thompson sampling algorithms from the analyses in this paper, as they are not strictly naive---Thompson sampling relies on instantiating a prior reward distribution. However, these greedy-in-the-limit algorithms enjoy practical use, and strong performance guarantees when properly initialized.

\paragraph{Function approximation algorithms.} A considerable amount of work studies collusion arising from agents choosing from a countable set of actions. In learning, these actions are treated as discrete, with no approximation of underlying demand functions across actions. While this work is important for laying the foundation of learning theory,  potentially fruitful future research could explore how learning outcomes change when actions are treated as being related by some metric.

\paragraph{Asynchronous updating.} Our theoretical model involves synchronous updating and learning. Further study of naive algorithmic collusion could benefit from studying how asynchronous action selection or updating might affect collusion potential, similar to the work of \citet{asker2022}.

\paragraph{Non-limiting behavior.} Our analysis studies the limiting behavior of learned outcomes by bandit agents. However, this limiting behavior does not say anything about the fraction of total plays where both agents charge supra-competitive prices. While short-term periods of supra-competitive pricing by firms may not pose serious harm in some markets where demand is flexible across time; in other markets, this alternative measure of collusion may matter.



\printbibliography






  


\section{Appendix}
\subsection{Additional Examples}\label{appendix:examples}
\subsubsection{Persistently random algorithms that collude}\label{ex:random_collude}
Consider a collusive algorithm $\mathcal{A}_C$. $\mathcal{A}_C$ is defined as follows:
\[
\Pr(a_{i,t}=H) = \begin{cases}
    .99 & \text{if } t \text{ mod } 2 = 0\\
    .51 & \text{if } t\text{ mod }2 = 1\\
\end{cases}.
\]
$\mathcal{A}_C$ is persistently random with $\epsilon = .01$ and also will adhere to rank-ordered play (i.e. always placing more mass on the higher-valued action) when $\pi_i^* = H$. Now, say both agents employ $\mathcal{A}_C$. We have $\mathbb{E}[r_{i,t}|a_{i,t} = H]= \beta\tfrac{(.99)^2+(.51)^2}{(.99)+(.51)}$ and $\mathbb{E}[r_{i,t}|a_{i,t} = L]= \gamma\tfrac{(.01)^2+(.49)^2}{(.01)+(.49)}+\tfrac{(.01)(.99)+(.49)(.51)}{(.01)+(.49)}$. Both agents employ this fixed algorithm, so each value converges to the expected reward. Here, with $\beta =.9, \quad \gamma = .4$, we have $V_{i,H,\infty} \approx .744 > V_{i,L,\infty} \approx .712$.
\subsubsection{Asymmetric UCB}
In Figure \ref{fig:asymmetric_ucb_fixed} we see a graph where game parameters $(\beta,\gamma)$ are fixed and $(\delta_1,\delta_2)$ varies along a lattice from $[\tfrac{1}{80},1]$ with increments of $\tfrac{1}{80}$.  Because each play path is deterministic, we only run a single trial for each, yielding $3,240$ trials in total. Each trial lasts $10,000$ rounds. Tan tiles indicate the corresponding  learned collusion and purple tiles indicate learned competition. We first note that along the diagonal in $(\delta_1,\delta_2)$, we see collusion emerge. This verifies the proof of Proposition \ref{prop:UCB}. We also note that collusion potential is highly sensitive to exploration parameters, with no clearly discernible pattern. 
\begin{figure}[h!]
    \centering
    \begin{subfigure}{0.3\textwidth}
        \includegraphics[width=\textwidth]{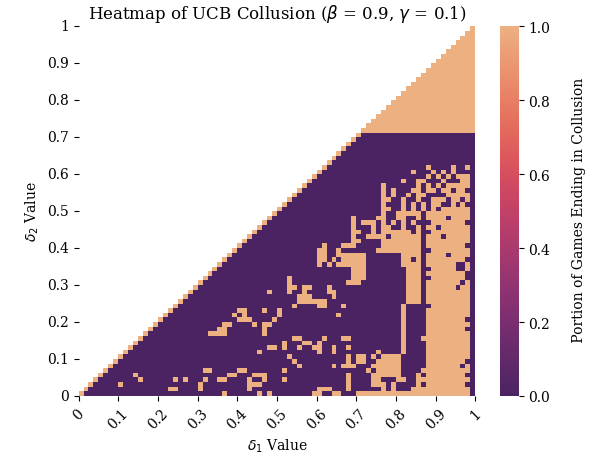}
        \subcaption{Heatmap of collusion in $(\delta_1,\delta_2)$ for $\beta = .9, \gamma = .1$. }
    \end{subfigure}
    \hfill
    \begin{subfigure}{0.3\textwidth}
        \includegraphics[width=\textwidth]{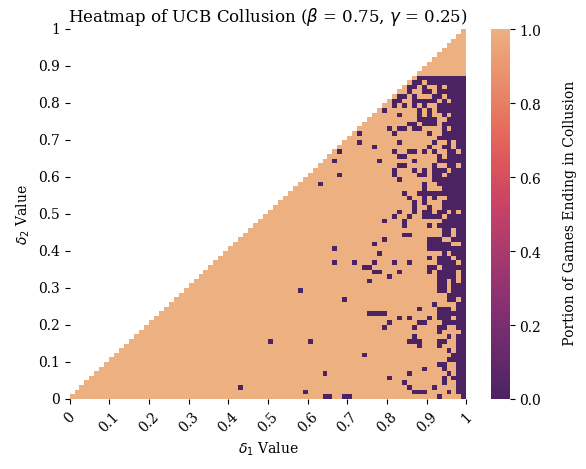}
        \subcaption{Heatmap of collusion in $(\delta_1,\delta_2)$ for $\beta = .75, \gamma = .25$. }
    \end{subfigure}
    \hfill
    \begin{subfigure}{0.3\textwidth}
        \includegraphics[width=\textwidth]{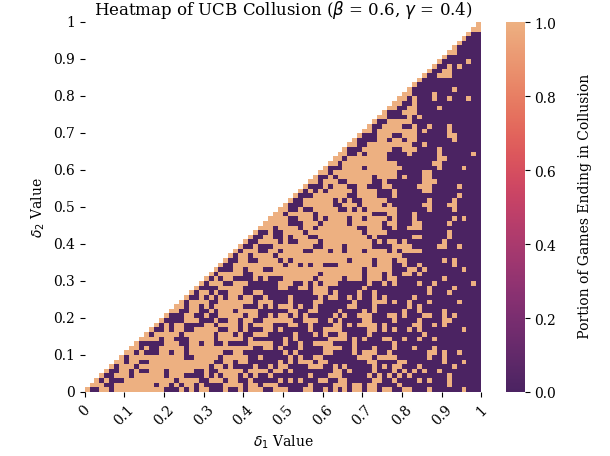}
        \subcaption{Heatmap of collusion in $(\delta_1,\delta_2)$ for $\beta = .6, \gamma = .4$. }
    \end{subfigure}
    \caption{Demonstration of $(\delta_1,\delta_2)$ values for fixed $(\beta,\gamma)$ where games end in collusion.}
    \label{fig:asymmetric_ucb_fixed}
\end{figure}

\subsection{Supplemental Proofs}
\label{appendix:proofs}
\subsubsection{Proof of Lemma \ref{lem:markov}}\label{prf:markov}
\vspace{0.5\baselineskip}
\noindent
\begin{proof}
There exists a bijective mapping from any outcome vector $o \in O$ to a deterministic reward vector $r$. Thus, for each index $i$ in the state $s_t$, there exists a mapping from $i$ to an outcome $o$ and a mapping from $o$ to some reward vector $r$. 

Start with zero-length reconstructed action histories $\alpha_{i,a}'$ for all actions $a \in A$ and a reconstructed reward history $\rho_i'$  across each player ${i}$. For each index $j$ in $[1,2^n]$, find the mapping from $j$ to $o$ and from $o$ to $r$. Now, for each player $i$, append $s_{t,j}$ copies of $r_i$ to  $\rho_i^{'}$. Additionally, append $s_{t,j}$ copies of 1 to the action vector $\alpha_{i,a'}'$ if $o_i = a'$ and $s_{t,j}$ copies of $0$ to  $\alpha_{i,a}'$ for $a \neq a'$. When all elements in $s_t$ have been iterated over, then for each player, the mapping of action played to reward will be preserved in $\alpha_{i,a}'$ and $\rho_i'$ by aligning indices for all players across all actions for all rounds until $t$. Thus $\alpha_{i,a} \rho_i = \alpha_{i,a}' \rho_i'$ and $\alpha_{i,a} \vec{1} =  \alpha_{i,a}' \vec{1} $.

Each bandit $i$ is imbued with an algorithm $\mathcal{A}_i$ that takes as an input the history $\mathcal{H}_i$ and outputs some probability distribution $\pi_i^t$ across actions. This strategy determines the distribution of actions played by each agent $i$. Because the probability of all actions $a \in A$ played by all agents $i$ is well defined by $s_t$, then these probabilities can be propagated to a distribution across all outcomes in $O$. With $s_t$  defined by the count of each outcome, the distribution over outcomes in $O$ has a bijective mapping to the distribution over states $s_{t+1}$. Thus, the current state alone is needed to determine the transition probabilities over successor states. Therefore, the state is Markov. 
\end{proof}
\subsubsection{Proof of Proposition \ref{prop:no-covariance}}\label{prf:no-covariance}
\vspace{0.5\baselineskip}
\noindent
\begin{proof}
In order for Player $i$ to observe $V_{i,H}(s_t)\geq V_{i,L}(s_t)$, $\beta \xi_{i,t,H} + (1-\gamma)\xi_{i,t,L} \geq 1$.  We have as $\beta \rightarrow 1, \gamma \rightarrow 0$, $\beta \xi_{i,t,H} + (1-\gamma)\xi_{i,t,L} \uparrow \xi_{i,t,H} +\xi_{i,t,L}$. As $\beta$ and $\gamma$ exist on open intervals,  we get the minimal conditions on $\xi$ that could lead to $V_{i,H}(s_t)\geq V_{i,L}(s_t)$ for any $\beta, \gamma$. For convenience, let $\tilde{a}_{i,t} =\mathds{1}
\{a_{i,t}=H\} $:
\[
\frac{\sum_{k=0}^t\tilde{a}_{i,k}\tilde{a}_{-i,k}}{\sum_{k=0}^t\tilde{a}_{i,k}} +\frac{\sum_{k=0}^t(1-\tilde{a}_{i,k})(1-\tilde{a}_{-i,k})}{\sum_{k=0}^t(1-\tilde{a}_{i,k})} > 1.
\]
      
      Rearranging, we arrive at:
      \[
t\sum_{k=0}^{t}\tilde{a}_{i,t}\tilde{a}_{-i,t}>\left(\sum_{k=0}^{t}\tilde{a}_{i,t}\right)\left(\sum_{k=0}^{t}\tilde{a}_{-i,t}\right).
      \]
Let $\bar{a}_{i} = \tfrac{1}{t}\sum_{k=0}^{t}\tilde{a}_{i,t}$. Dividing through by $t^2$ and rearranging, we get:
\begin{align*}
    \frac{1}{t}\sum_{k=0}^{t}\tilde{a}_{i,t}\tilde{a}_{-i,t}-\bar{a}_{i}\bar{a}_{-i}=: Cov(a_{i},a_{-i}) > 0.
\end{align*} 
\end{proof}

\subsubsection{Proof of Proposition \ref{prop:epsilon-greedy}}\label{prf:epsilon-greedy}
\vspace{0.5\baselineskip}
\noindent
\begin{proof}
With epsilon-greedy, there are four possible regimes that dictate the behavior of agents. We define a regime as $\Pi^* = (\pi^*_1, \pi^*_2)$. We show that in each of these possible regimes, value estimates are pushed away from those that induce the collusive regime $\Pi^* = (H,H)$. We analyze the limit behavior of the value estimates of Player $1$; however the results generalize symmetrically to Player $2$.

Assume agents are in the collusive regime $\Pi^{*} = (H,H)$. Let $\tilde{a}_{i,t}$ be the Bernoulli random variable corresponding to $1$ if $a_{i,t} = H$ and $0$ if $a_{i,t} = L$. With this construction, $\text{P}(\tilde{a}_{i,t} = 1) = 1-\frac{\epsilon_i}{2}$ at any time $t$ while agents are in the collusive regime. Now, from time $t$ to some time $t + t' > t$,  we have the following value estimates for Player 1 defined by these Bernoulli trials:

\begin{align*}
    V_{1,H}(s_{t + t'}) &=  \frac{\beta (s_{t,1}+\sum_{k=1}^{t'}\tilde{a}_{1,k}\tilde{a}_{2,k})}{s_{t,1}+s_{t,2}+\sum_{k=1}^{t'}\tilde{a}_{1,k}},\\
    V_{1,L}(s_{t+t'}) &= \frac{s_{t,2}+\sum_{k=1}^{t'}(\tilde{a}_{1,k})(1-\tilde{a}_{1,k})}{s_{t,2}+s_{t,3} + \sum_{i=1}^{t'}(1-\tilde{a}_{1,k})}\\
    &+\frac{\gamma( s_{t,3} + \sum_{i=1}^{t'}(1-\tilde{a}_{1,k})(1-\tilde{a}_{1,k}))}{s_{t,2}+s_{t,3} + \sum_{i=1}^{t'}(1-\tilde{a}_{1,k})}.
\end{align*}

As $t' \rightarrow \infty$, $V_{1,H}(s_{t + t'}) = \beta(1-\frac{\epsilon_2}{2})$ and, $V_{1,L}(s_{t + t'}) = 1-\frac{\epsilon_2}{2} + \gamma\frac{\epsilon_2}{2}$ by the law of large numbers. Thus, as $t' \rightarrow \infty$, $V_{1,H}(s_{t + t'}) < V_{1,L}(s_{t + t'})$ if players are in $\Pi^{*} = (H,H)$. So, Player $1$ will eventually change its target policy to $L$, pushing the agents out of the collusive regime.

Now, if agents are in $\Pi^{*} = (L,H)$, we observe that Player $1$ will in expectation stay in $\pi^{*}_1 = L$. In this regime, we now have $\text{P}(\tilde{a}_{1,t} = 1) = \frac{\epsilon_1}{2}$. We still have the same limit value as $t' \rightarrow \infty$, $V_{1,H}(s_{t + t'}) = \beta(1-\frac{\epsilon_2}{2})$ and $V_{1,L}(s_{t + t'}) = 1-\frac{\epsilon_2}{2} + \gamma\frac{\epsilon_2}{2}$. In the competitive regime, $\Pi^{*} = (L,L)$, $\tilde{a}_{1,t}$ is now defined such that $\text{P}(\tilde{a}_{1,t} = 1) = \frac{\epsilon_1}{2}$. This leads to the limit as $t' \rightarrow \infty$ of $V_{1,H}(s_{t + t'}) = \beta\frac{\epsilon_2}{2}$ and $V_{1,L}(s_{t + t'}) = \frac{\epsilon_2}{2} + \gamma(1-\frac{\epsilon_2}{2})$. Finally, for, $\Pi^{*} = (H,L)$,  $\text{P}(\tilde{a}_{1,t} = 1) = 1 - \frac{\epsilon_1}{2}$ and $\text{P}(\tilde{a}_{1,t} = 1) = \frac{\epsilon_1}{2}$. This leads to the same limit as that of $\Pi^* = (L,L)$. In all three regimes, as $t' \rightarrow \infty$, $V_{1,H}(s_{t + t'}) < V_{1,L}(s_{t + t'})$. Thus, the value estimates for Player $1$ and Player $2$, by the symmetry of this argument, will be pushed away from those that induce $\pi^{*} = (H,H)$, regardless of the current regime. Thus, all regimes but $(L,L)$ are transient, leading agents to converge to play of $\pi_{i,\infty}(H) = \tfrac{\epsilon_i}{2}$.

\end{proof}
\subsubsection{Proof of Proposition \ref{prop:convergence-competition}}\label{prf:convergence-competition}
\vspace{0.5\baselineskip}
\noindent
\begin{proof}
By persistent randomness, $N_{i,a}(t) \rightarrow \infty$. Thus, $\bar{s} = \tfrac{s_t}{t} \rightarrow (p_1^\star p_2^\star, p_1^\star (1-p_2^\star), (1-p_1^\star) p_2^\star, (1-p_1^\star)(1- p_2^\star))$.  So, we have 
    \begin{align*}
        V_{i,H}(\bar{s}) &= \beta p_{-i}^\star,\\
        V_{i,L}(\bar{s})& = \gamma+ (1-\gamma)p_{-i}^\star.
    \end{align*}
    As we multiply $\bar{s}$ by $t$, we have $V_{i,a}(s_\infty) = V_{i,a}(\bar{s})$. By the game payoff construction $1>\beta>\gamma>0$, we have that $V_{i,L}(\bar{s})>V_{i,H}(\bar{s})$, which yields the result.
\end{proof}
\subsubsection{Proof of Proposition \ref{prop:etc-exp}}\label{prf:etc-exp}
\vspace{0.5\baselineskip}
\noindent

We show this bound holds via concentration inequalities.
\begin{proof}
First, we define the set of events where $s_{t^*}$ has sufficient coverage across all dimensions to induce a small Lipschitz constant on $\rho_{i,a}$.  To do this, define ``good" events $G = \{s \in S_{t^*}, s_1+s_2 \geq \frac{t^*}{8}, s_1+s_3 \geq \frac{t^*}{8}, s_3+s_4 \geq \frac{t^*}{8}, s_2+s_4 \geq \frac{t^*}{8}\}$ . Now, we will bound the probability that $s_{t^*}$ falls outside of $G$.

Now, the sum of any two  $s_i + s_j \sim Binomial(\frac{1}{2}, t^*)$. Following the Chernoff bound for binomials, we have $P(s_i + s_j \leq (1-\delta)\frac{t^*}{2})\leq \text{exp}(-\frac{\delta^2t^*}{2})$. Setting $\delta = 3/4$, we have $P(s_i + s_j \leq \frac{t^*}{8})\leq \text{exp}(-\frac{9t^*}{32})$ for a single $i \neq j$. We then use a union bound to bound each of these $4$ conditions so that $\mathbb{P}(s_{t^*} \notin G) \leq 4 \text{exp}(-\frac{9t^*}{32})$.

Now, we can condition on $s_{t^*}$ having sufficient coverage to find a Lipschitz constant that decreases in $t^*$. We will use this constant to apply McDiarmid's inequality on $\Pr(s_{t*} \in C|\beta, \gamma, t^*, s_{t^*}\in G)$.

Let $f(x_1,x_2)=\frac{x_1}{x_1+x_2}$ be restricted to the domain where $x_1+x_2 \geq \frac{t^*}{8}$. Now, $\partial{x_1} f = \frac{x_2}{(x_1+x_2)^2}$ and $\partial{x_2}f = \frac{-x_1}{(x_1+x_2)^2}$. These partial derivatives imply $|\partial f|\leq\frac{1}{x_1+x_2}\leq \frac{8}{t^*}$, giving us a Lipschitz constant of $\frac{8}{t^*}$. This bounds how much $\rho_i$ can change, given that the denominator has sufficient coverage.

Let $g_i = \beta \rho_{i}(H) + (1-\gamma)\rho_i(L)$, where $C$ is defined by $g_i > 1 \;\forall \;i$. We have $\mathbb{E}[g_i] = \frac{1+\beta-\gamma}{2}$, giving a ``gap" in $g_i$ of $\Delta = \frac{1-\beta+\gamma}{2}$ by which all $g_i$ must deviate from their mean to reach the collusive state. We will use McDiarmid's inequality to bound the probability that $g_i$ deviates from its mean by greater than $\Delta$. Generally, we construct $g_1(o_1,...,o_t,...,o_{t^*}) = \beta f(\sum_{t=1}^{t^*}o_{i,1},\sum_{t=1}^{t^*}o_{i,2})-(1-\gamma)f(\sum_{t=1}^{t^*}o_{i,4},\sum_{t=1}^{t^*}o_{i,3})$ with $o_t \in \mathbb{R}^4$. Changing $o_t$ to some other outcome $o_t'$ will result in a bounded change to $g_1$ such that $|g_1(o_1,...,o_t,...,o_{t^*})-g_1(o_1,...,o_t',...,o_{t^*})| \leq\frac{8(1 + \beta - \gamma)}{t^*}$ due to the Lipschitz continuity of $f$. The same bound holds for $g_2$, which can be shown by changing the inputs to $f$. Plugging this in to McDiarmid's inequality, we get $\mathbb{P}(g_i-\mathbb{E}[g_i] \geq \Delta)\leq \text{exp}(\frac{-2\Delta^2}{t^*(\frac{8(1 + \beta - \gamma)}{t^*})^2}) =\text{exp}(-(\frac{(1-\beta+\gamma)^2}{128(1 + \beta - \gamma)^2})t^*) $. We now use this to upper bound this probability, conditioned  on a state being ``good
by $\mathbb{P}(s_{t*} \in C|\beta, \gamma, t^*, s_{t^*}\in G) = \mathbb{P}((g_1-\mathbb{E}[g_1] \geq \Delta) \land (g_2-\mathbb{E}[g_2] \geq \Delta) )\leq \mathbb{P}(g_i-\mathbb{E}[g_i] \geq \Delta) \leq \text{exp}(-(\frac{(1-\beta+\gamma)^2}{128(1 + \beta - \gamma)^2})t^*) $. Finally, we have $\mathbb{P}(s_{t*} \in C|\beta, \gamma, t^*) \leq \mathbb{P}(s_{t*} \in C|\beta, \gamma, t^*, s_{t^*}\in G) + \mathbb{P}(s_{t^*} \notin G) \leq \text{exp}(-(\frac{(1-\beta+\gamma)^2}{128(1 + \beta - \gamma)^2})t^*)+4 \text{exp}(-\frac{9t^*}{32})\leq c_1e^{-t^*c_2}$. 
\end{proof}

Now, we approximate this process via a Gaussian approximation of the multinomial and the delta method. By the central limit theorem, the multinomial converges in distribution to a multivariate Gaussian. We approximate $\sqrt{t^*}(s_{t^*} -\mu t^*)\sim \mathcal{N}(0, t^*\Sigma)$, via the Central Limit Theorem with $\mu_i =  p_i =.25$ and $\Sigma_{i,i} = p_i(1-p_i) = \frac{3}{16}$ and $\Sigma_{i,j} = -p_ip_j = -\frac{1}{16}$. This approximation is standard, as described in \citet{casella2002statistical}.

The delta method states $\sqrt{t^*}(g(s_{t^*} )-g(\mu)) \xrightarrow[]{d} \mathcal{N}(0, J(\mu)\Sigma J(\mu)^T)$ \citep{casella2002statistical}. Now, let  $g(s) = h(\xi(s))$, with $h$ constructed in the following way:
\[
  h(\xi) = 
  \begin{pmatrix}
    \beta \xi_{1,H} + (1-\gamma)\xi_{1,L}\\
    \beta \xi_{2,H} + (1-\gamma)\xi_{2,L}
  \end{pmatrix}
\]
This gives us the following Jacobian with respect to state $s$:

\begin{equation*}
\begin{gathered}
  J(s) = \frac{\partial g_i }{\partial s_{i}}=  \\
  \left(
  \begin{smallmatrix}
    \beta\frac{s_{2}}{(s_{1}+s_{2})^2}&-\beta\frac{s_{1}}{(s_{1}+s_{2})^2}&-(1 - \gamma )\frac{s_{4}}{(s_{3}+s_{4})^2}& (1 - \gamma)\frac{s_{3}}{(s_{3}+s_{4})}   \\
    \beta\frac{s_{3}}{(s_{1}+s_{3})^2}&-(1-\gamma)\frac{s_{4}}{(s_{2}+s_{4})^2}&-\beta\frac{s_{1}}{(s_{1}+s_{3})^2}& (1 - \gamma)\frac{s_{2}}{(s_{2}+s_{4})^2}  
  \end{smallmatrix}
  \right).
\end{gathered}
\end{equation*}

When we evaluate this at $s_{t^*} = \mu$, we get:
\[
  J(\mu)=
  \begin{pmatrix}
\beta &-\beta&-(1-\gamma )& (1 - \gamma)   \\
\beta&-(1-\gamma )& -\beta&(1 - \gamma)
  \end{pmatrix}
\]
This gives us:
\[
  J(\mu)\Sigma J(\mu)^T=
  \begin{pmatrix}
\frac{(1-\gamma)^2+\beta^2}{2}&   \frac{(1+\beta-\gamma)^2}{4}\\
\frac{(1+\beta-\gamma)^2}{4}&\frac{(1-\gamma)^2+\beta^2}{2}
  \end{pmatrix}
\]
We have that the $g_i(\mu)=\frac{1}{2}(\beta+1-\gamma)$.  Now, dividing by $\sqrt{t^*}$  and adding $g(\mu)$, we arrive at 
\[
(g_{1}(s),g_{2}(s))
  \;\sim\;
  \mathcal N\!\Bigl(
      (m,m),
      \begin{pmatrix}
        \sigma^{2} & \psi\,\sigma^{2}\\
        \psi\,\sigma^{2} & \sigma^{2}
      \end{pmatrix}
    \Bigr),
\]
with
\begin{equation*}
\begin{gathered}
m=\tfrac12(\beta+1-\gamma),\\
\sigma^{2}=\frac{\beta^{2}+(1-\gamma)^{2}}{2t^*},\\
\psi=\frac{(1+\beta-\gamma)^{2}}{2\,[\beta^{2}+(1-\gamma)^{2}]}.
\end{gathered}
\end{equation*}

Now, we find $\Pr(g_1>1, g_2 > 1)$ by defining the one–sided \(z\)-score for some $Z_i \sim \mathcal{N}(0,1)$
\[
z \;=\; \frac{1-m}{\sigma}
       \;=\;
       \sqrt{t^*}\,
       \frac{1-\beta+\gamma}
            {\sqrt{2\,[\beta^{2}+(1-\gamma)^{2}]}}
       \;>\;0 .
\]

We now draw on large-deviation theory to get a clean approximate expression of $\Pr(Z_1 > z, Z_2 > z)$.  \citet{dembo1998} gives the following:
\[
-\lim_{t\to\infty}\frac{1}{t^{2}}\log\Pr\!\bigl((Z_1,Z_2)\in tB\bigr)
   =\inf_{x\in B}I(x),
\]
with rate
\[
I(x)= \frac{1}{2}x^T\Sigma^{-1} x = \frac{x_{1}^{2}-2\psi x_{1}x_{2}+x_{2}^{2}}{2(1-\psi^{2})}.
\]
We rearrange to get our approximation of interest:
\[
\Pr((Z_1,Z_2) \in tB) \approx \text{exp}\{-t^2 \inf_{x\in B}I(x)\}.
\]

Setting $B=[1,\infty)^{2}$,  and using convexity of $I$,
the infimum is attained at the corner $(1,1)$. Then, scaling by $t=z$ yields
$
\inf_{x\in A_{z}}I(x)=\frac{1}{1+\psi}.
$
This gives us
\[
\Pr(Z_1 > z, Z_2>z) \approx \text{exp}\{-\frac{z^2}{1+\psi} \}.
\]
We now apply this result to $(g_{1}(s),g_{2}(s))$. 
Substituting \(z^{2}\) and \(\rho\) into gives the closed-form exponential approximation
\begin{equation}
\label{eq:joint-tail-proxy}
\Pr\!\bigl(g_{1}(s)>1,\;g_{2}(s)>1\bigr)
   \;\approx\;
   \exp\!\{-t^*\,C(\beta,\gamma)\},
\end{equation}
with

\[
C(\beta,\gamma)=
\frac{(1-\beta+\gamma)^{2}}
     {\,2\bigl[\beta^{2}+(1-\gamma)^{2}\bigr]
      +(1+\beta-\gamma)^{2}}.
\]

\subsubsection{Proof of Proposition \ref{prop:ed-beta}}\label{prf:ed-beta}
\vspace{0.5\baselineskip}
\noindent
\begin{proof}
    
The approximation is predicated on agents beginning in the regime where $\Pi^*= (L,L)$ (i.e. agents both observe playing $L$ to be optimal) for early play, which we now justify.

For $\eta$ near $1$ and at small $t$, agents explore approximately uniformly. In this early gameplay, $\pi_i(a) \approx \frac{1}{2}$ for $a \in \{H,L\}$, regardless of which action is greedy for each agent. This will lead to approximately the same distribution of $\xi_{i,a}$ as in the case of ETC. Assuming $\pi_i(a) = \frac{1}{2}$, this distribution concentrates around $\xi_{i,a}=\frac{1}{2}$ with variance decreasing on the order of $\frac{1}{t}$.

When $\pi_i(H)-\pi_i(L)$ starts to diverge, play paths will vary based on the regime $\Pi^*$. However, we know that $\Pi^*=(L,L)$ is exceedingly likely in early rounds given that $\xi_{i,a} = \frac{1}{2}$ is in the regime $\Pi^*=(L,L)$. In words, agents start off with near-uniform random sampling, ensuring that probability of agents synchronizing on action plays is roughly $.5$ for each action. This leads to near-zero correlation in action plays, where agents would observe $L$ to be the higher-valued action. Going forward, we restrict our analysis to Player 1 (i.e., $\Pr(\bar{s} \in C) \approx \Pr(V_{1,H}(\bar{s}) > V_{1,L}(\bar{s}))$. This is a safe approximation as $\bar{s}_2 - \bar{s}_3$ will deviate by only small amounts in the limit with play in $\Pi^* = (L,L)$ as $\bar{s}$ would converge to be degenerate on the index corresponding to the outcome $(L,L)$. 

In $\Pi^* = (L,L)$, we have $\xi_{1,L} = \frac{s_4}{s_3 + s_4}\rightarrow  \frac{\sum_{t=0}^{\infty}(1-\frac{\eta^t}{2})^2}{\sum_{t=0}^{\infty}[(1-\frac{\eta^t}{2})^2 +(1-\frac{\eta^t}{2})\frac{\eta^t}{2}]} =1$. This occurs as, if agents stayed in $\Pi^* = (L,L)$, the fraction of plays on arm $L$ would converges strictly to play of $L$ such that  $o_\infty = (L,L)$ almost surely. Now, $\xi_{1,H}$ does not have the same convergence guarantee. $\xi_{1,H}$ depends on play early in gameplay as ultimately, if play were to stay in $\Pi^*=(L,L)$, the cumulative probability of action plays would be concentrated in early periods due to the geometrically decaying exploration rate. Instead, we can model this value as $t \rightarrow \infty$ as a Beta distribution.  Below, we show how this approximation is devised.

Define rates
\begin{align*}
\lambda_1
   \;=\;\mathbb{E}[\bar{s}_1]
   &=\sum_{t=0}^{\infty}\frac{\eta^{2t}}{4}
     =\frac{1}{4(1-\eta^{2})},\\
\lambda_2
   \;=\;\mathbb{E}[\bar{s}_2]
   &=\sum_{t=0}^{\infty}\Bigl((1-\tfrac{\eta^t}{2})\tfrac{\eta^{t}}{2}\Bigr)\\
    & =\frac{1}{2(1-\eta)}-\lambda_1=\frac{1+2\eta}{4(1-\eta^2)}.
\end{align*}

The count of any type of outcome $o_t$ in this regime could be modeled as a Poisson-Binomial distribution, which is notoriously tricky to work with. However, since the events $(o_t=(H,H))$ and $(o_t=(H,L))$ are rare (with rates decreasing to $0$) and independent across $t$, the Poisson limit theorem gives counts of $s_1, s_2$ as the approximation
\[
s_1\sim\text{Poisson}\!\bigl(\lambda_1\bigr),
\qquad
s_2\sim\text{Poisson}\!\bigl(\lambda_2\bigr).
\]

Given $n=s_1+s_2$, we now have the conditional probability
\[
s_1\,\bigl|\,n \;\sim\; \text{Binomial}\!\Bigl(n,\,
           p=\frac{\lambda_1}{\lambda_1+\lambda_2}\Bigr).
\]

Marginalizing over $n$ yields the continuous approximation
\[
\xi_{1,H}(s)\;\sim\;
\text{Beta}\bigl(\lambda_1,\,\lambda_2\bigr),
\]
valid whenever $\lambda_1+\lambda_2$ is not too small, which happens as $\eta$ is close to $1$.

Now, to evaluate $\Pr(\bar{s}_\infty \in C |\beta, \gamma, \eta) $, we are interested in when
\[
\beta \xi_{1,H} + (1-\gamma) \xi_{1,L} > 1,
\]
which, with $t \rightarrow \infty$ ensuring $\xi_{1,L} = 1$ , gives us the requirement  $\xi_{1,H} > \frac{\gamma}{\beta}$ for Player 1 to observe $V_{1,H}(s_\infty) > V_{1,L}(s_\infty)$. 

Combining the steps above, we obtain the compact formula
\[
\Pr(\xi_{1,H}(s_\infty)>\tfrac{\gamma}{\beta})
\;\approx\;
1-\mathrm{B}\!\bigl(\tfrac{\gamma}{\beta};\;
                    \lambda_1,\lambda_2\bigr)    
\]
with
\[
\lambda_1=\frac{1}{4(1-\eta^{2})}, \qquad 
\lambda_2=\frac{1+2\eta}{4(1-\eta^2)}.
\]
Here, $\mathrm{B}(x; \cdot)$ denotes the cumulative distribution
function of a Beta random variable with the
parameters $\lambda_1,\lambda_2$. 
\end{proof}
\subsubsection{Proof of Lemma \ref{lem:monotonicity-H}}\label{prf:monotonicity-H}
\vspace{0.5\baselineskip}
\noindent
\begin{proof}
Observe that if $o_t = (H,H)$, then in the next round, we have $s_{t+1,1} = s_{t,1} + 1$. Also, note that no entry in $s_t$ may be negative. From the construction of value estimates $V$, we can perform simple algebraic manipulation to arrive at a proof of the claim.

We know $V_{i,H}(s_t)$ to be a linear combination of $0$ and $\beta$, such that $0\leq V_{i,H}(s_t)\leq \beta$. Without loss of generality, we prove the results for Player $1$. In case 1, assume that $s_{t,2} > 0$. Then, via simple algebraic manipulation, we arrive at the conclusion. 

    \begin{align*}
        s_{t,2}&>0\\
        \frac{s_{t+1,1}\beta}{s_{t+1,1} + s_{t,2}} &> \frac{s_{t,1}\beta}{s_{t,1} + s_{t,2}}\\
        V_{1,H}(s_{t+1}) &> V_{1,H}(s_{t})
    \end{align*}

In case 2, we assume $s_{t,2} = 0$ and observe the following.

    \begin{align*}
        s_{t,2}&=0\\
        \frac{s_{t+1,1}\beta}{s_{t+1,1} + s_{t,2}} &= \frac{s_{t,1}\beta}{s_{t,1} + s_{t,2}}\\
        V_{1,H}(s_{t+1}) &= V_{1,H}(s_t) = \beta
    \end{align*}

The proof for both cases above can be applied to a general Player $i$ in an $n$-player Prisoner's Dilemma by substituting the sum across all state values corresponding to outcomes with $o_{t,i} = H$ and $o_t \neq \vec{H}$ for $s_{t,2}$.
\end{proof}
\subsubsection{Proof of Lemma \ref{lem:monotonicity-L}}\label{prf:monotonicity-L}
\vspace{0.5\baselineskip}
\noindent
\begin{proof}

We follow a near identical procedure to the proof of Lemma \ref{lem:monotonicity-H} to this proof, so we omit the steps for readability.

Observe that if $o_t = (L,L)$, then in the next round, we have $s_{t+1,4} = s_{t,4} + 1$. 

We know $V_{i,L}(s_t)$ to be a linear combination of $\gamma$ and $1$, such that $\gamma \leq V_{i,L}(s_t)\leq 1$. Without loss of generality, we prove the results for Player $1$. In case 1, assume that $s_{t,3} > 0$ and in case 2, we assume $s_{t,3} = 0$. By algebraic manipulation, as in Lemma 2, we arrive at the result

\end{proof}
\subsubsection{Proof of Lemma \ref{lem:subsequent}}
\vspace{0.5\baselineskip}
\noindent
\begin{proof}

Let $\mathcal{H}_{i,t}$ and $\mathcal{H}_{j,t}$ be the path-equivalent histories, and let $\mathcal{H}_{t}$ be some path-equivalent ordering of both $\mathcal{H}_{i,t}$ and $\mathcal{H}_{j,t}$. Denote the player's symmetric algorithm as $ \mathcal{A}_i$=  $\mathcal{A}_1 = \mathcal{A}_2$. Since the algorithm is deterministic, at any time $t$, we have for some action $a$, 
\begin{align}
\Pr( \mathcal{A}(\mathcal{H}_{t}) = a) =  1,
\end{align}
which will be the action  $t = t^{*}$. Thus, there are only two outcomes in this period. Players may play $H$, $o_t = (H,H)$, which would imply $r_t = (\beta,\beta)$. Otherwise, players play $L$, $o_t = (L,L)$, leading to $r_t = (\gamma,\gamma)$. These outcomes assign the same reward to all players (call this $r_{i,t}$ ). All players will then apply the same update to their history of rewards, such that $\forall \; i, \; \; \rho_{i,t+1} = \rho_{i,t} \frown r_{t}$.\footnote{Here, $\frown$ refers to vector concatenation.} The history of action plays will similarly be updated symmetrically across action plays, with a $1$ appended to the action-play history of the played action and a $0$ is appended to the action-play history of the unplayed action. If in round $t$, histories were path-equivalent between players, then with the concatenation of the same values to all vectors in each player's history, then histories will be path-equivalent between players in round $t+1$ . Thus, the induction holds.

Now, with histories being path-equivalent across all players, we have $\mathcal{A}_i(\mathcal{H}_{i,t}) = \pi_t$, where $\pi_t$ is the same degenerate distribution across actions for each player $i$. Therefore, in any round $t > t^*$, all players will play the same action with probability $1$. The result follows. 
\end{proof}
\subsubsection{Proof of Proposition \ref{prop:deterministic}}\label{prf:deterministic}
\vspace{0.5\baselineskip}
\noindent
\begin{proof}
By Lemma \ref{lem:subsequent}, we know that when $\mathcal{H}_{1,t}$ and $\mathcal{H}_{2,t}$ are path-equivalent, players will play the same action in periods $t . . . \infty$. Now, in round $t=0$, all vectors in each agent's history are empty. Thus, for all $t\geq 0$, players will symmetrically play $H$ or $L$. This leads to outcomes such that  $o_t = \vec{H}$ implying $r_t = \vec{\beta}$ or $o_t = \vec{L}$, implying $r_t = \vec{\gamma}$ .
    
    Let period $T$ be the first period where both actions have been sampled. In the two-player case, this means $s_{T,1} > 0$ and  $s_{T,4} > 0$. Both players now  have value estimates for $H$ and $L$ that are different from their initialized values, and $V_{i,H}(s_{T}) = \beta > \gamma = V_{i,L}(s_{T})$.
    
    Now,  we know that for all $t \geq T$, $o_t = \vec{H}$ or $o_t = \vec{L}$. Thus, by Lemma 4,  for all $t \geq T, \forall \; i, \; V_{i,H}(s_t) = \beta > \gamma = V_{i,L}(s_t)$, which implies that $\pi_{i,t}^{*} = H$ for all $t \geq T$ .
\end{proof}
\subsubsection{Proof of Proposition \ref{prop:UCB}}\label{prf:UCB}
\vspace{0.5\baselineskip}
\noindent
\begin{proof}

When $t=0$, player UCB estimates are equal with $\text{UCB}_{1,H}(s_0)=\text{UCB}_{1,L}(s_0)=\text{UCB}_{2,H}(s_0)=\text{UCB}_{2,L}(s_0)=\infty$ by the algorithm construction. Thus, there are multiple maximal values. 

Observe that there are four possible outcomes in round $t=0$. The possible play paths are below:
    \begin{enumerate}
        \item $o_1 = (H,H)$. This leads to state $s_1 = (1,0,0,0)$ and $\text{UCB}_{1,H}(s_1)=\text{UCB}_{2,H}(s_1) = \beta + \sqrt{2 log(1/\delta)} < \infty$. Thus, in period $t=1$, players will choose $L$, with $\text{UCB}_{1,L}(s_1) = \text{UCB}_{2,L}(s_1) = \infty$, leading to $o_2 = (L,L)$ and $s_2 = (1,0,0,1)$.

        \item $o_1 = (L,L)$. This leads to state $s_1 = (0,0,0,1)$ and $\text{UCB}_{1,L}(s_1)=\text{UCB}_{2,L}(s_1) = \gamma + \sqrt{2 log(1/\delta)} < \infty$. Thus, in period $t=1$, players will choose $H$, with $\text{UCB}_{1,H}(s_1) = \text{UCB}_{2,H}(s_1) = \infty$, leading to $o_2 = (H,H)$ and $s_2 = (1,0,0,1)$.
        
        \item $o_1 = (H,L)$. This leads to state $s_1 = (0,1,0,0)$. Thus, in period $t=1$, Player $1$ will choose $a_{1,1} = L$ as $\text{UCB}_{1,H}(s_1) = \sqrt{2 log(1/\delta)} < \infty = \text{UCB}_{1,L}(s_1)$. Likewise, Player $2$ will observe $\text{UCB}_{2,L}(s_1) = 1 + \sqrt{2 log(1/\delta)} < \infty = \text{UCB}_{2,H}(s_1)$,leading to $a_{1,1} = H$. These choices result in $o_2 = (L,H)$ and $s_2 = (0,1,1,0)$.
        
        \item $o_1 = (L,H)$. This leads to state $s_1 = (0,0,1,0)$. Thus, in period $t=1$, Player $1$ will choose $a_{1,1} = H$ as $\text{UCB}_{1,L}(s_1) = 1 + \sqrt{2 log(1/\delta)} < \infty = \text{UCB}_{1,H}(s_1)$. Likewise, Player $2$ will observe $\text{UCB}_{2,H}(s_1) = \sqrt{2 log(1/\delta)} < \infty = \text{UCB}_{2,L}(s_1)$,leading to $a_{1,1} = L$. These choices result in $o_2 = (H,L)$ and $s_2 = (0,1,1,0)$.
    \end{enumerate}

Now, at $t = 2$, we have two possibilities of state: 

\begin{enumerate}
    \item In cases $1.$ and $2.$ above, we have $s_3 = (1,0,0,1)$. Thus, histories for each player are equal such that $\alpha_{1,H,2}\vec{\mathds{1}} = \alpha_{2,H,2}\vec{\mathds{1}} = 1$ and $\alpha_{1,L,2}\vec{\mathds{1}} = \alpha_{2,L,2}\vec{\mathds{1}} = 1$. Moreover, $\alpha_{1,H,2}\rho_{1,2} = \alpha_{2,H,2}\rho_{2,2} = \beta$ and $\alpha_{1,L,2}\rho_{1,2} = \alpha_{2,L,2}\rho_{2,2} = \gamma$. These equalities imply the path equivalence of $\mathcal{H}_{1,2}$ and $\mathcal{H}_{2,2}$. Now, with this path-equivalence, we have $\text{UCB}_{1,H}(s_3) = \text{UCB}_{2,H}(s_3) = \beta + \sqrt{2log(1/\delta)} > \gamma + \sqrt{2log(1/\delta)} = \text{UCB}_{1,L}(s_3) = UCB_2(L,s_3) $. Thus, agents no longer observe the same UCB estimates between actions.

    \item  In cases $3.$ and $4.$ above, we have $s_3 = (0,1,1,0)$. This leads to a similar path-equivalence of histories between agents. Specifically, $\alpha_{1,H,2}\vec{\mathds{1}} = \alpha_{2,H,2}\vec{\mathds{1}} = 1$ and $\alpha_{1,L,2}\vec{\mathds{1}} = \alpha_{2,L,2}\vec{\mathds{1}} = 1$. Meanwhile, $\alpha_{1,H,2}\rho_{1,2} = \alpha_{2,H,2}\rho_{2,2} = 0$ and $\alpha_{1,L,2}\rho_{1,2} = \alpha_{2,L,2}\rho_{2,2} = 1$. These equalities similarly imply the path-equivalence of $\mathcal{H}_{1,2}$ and $\mathcal{H}_{2,2}$.
    Again from this path-equivalence, we observe similar symmetric values between agents yet unequal values between UCB values of actions for each agent. $\text{UCB}_{1,H}(s_3) = \text{UCB}_{2,H}(s_3) = 0 + \sqrt{2log(1/\delta)} < 1 + \sqrt{2log(1/\delta)} = \text{UCB}_{1,L}(s_3) = \text{UCB}_{2,L}(s_3) $. 

\end{enumerate}

Now, the \textit{argmax} function returns a single action and no longer relies on tie-breaking rules. Invoking Lemma 4, we see that for all $t\geq 2$, agents will choose the same actions such that $o_t = (H,H)$ or $o_t = (L,L)$. In cases 1. and 2., for all $t \in 2 ... \infty$, we have that $V_{i,H}(s_t) = \beta$  by Lemma 2 and $V_{i,L}(s_t) = \gamma$ by Lemma 3. Thus, in these cases, bandit agents will always observe $\pi^*_i =H$ for $t \geq 2$. 

However, in cases 3. and 4., there is no longer a strict guarantee that UCB bandits will converge for all $\delta$ values given $\beta$ and $\gamma$ values. With $s_3 = (0,1,1,0)$, we have $V_{i,H}(s_3) = 0 <1 = V_{i,L}(s_3)$.  Lemmas 2 and 3 guarantee that if each action is played infinitely, they will converge to $\beta$ and $\gamma$ respectively. However, there is no guarantee that $H$ will necessarily be sampled again by agents (i.e. it is possible for all $t\geq 2$, $o_t = (L,L)$ ).  Therefore, we need some restriction on $\delta$ such that $H$ will be sampled enough by agents so that its value estimate can surpass that of $L$. This restriction takes the form of an upper bound where $\delta < e^{\frac{-\gamma^2}{2}}$, if the logarithm implemented is the natural log. We prove this next in the appendix. With this restriction, agents will play $H$ until $V_{i,H}(s_t) >V_{i,L}(s_t) > \gamma$, so that by Lemmas 2, 3, and 4, agents will always observe $\pi^*_i =H$ for $t \geq 2$.


\end{proof}
\subsubsection{Proof of upper bound on $\delta$ for cases 3. and 4. in Proposition 
\ref{prop:UCB}}\label{prf:UCB-upper-bound}
\vspace{0.5\baselineskip}
\noindent
\begin{proof}

Observe that when $s_2 = (0,1,1,0)$, we have $\text{UCB}_{i,H}( s_2) = \sqrt{2log(1/\delta)}$.  Now assume $\delta< e^{\frac{-\gamma^2}{2}}$.  By Lemma 4, we know that for all $t \geq 2$, play of $H$ corresponds to $s_{t+1,1} = s_{t,1} +1$ and play of $L$ corresponds to $s_{t+1,4} = s_{t,4} +1$. These are the only two possible movements in state. If $s_{t,4} \rightarrow \infty$, then $\text{UCB}_{i,L}(s_t)$ approaches $\gamma$. So, for $H$ to be sampled at least once more in this infinitely repeated game, $\sqrt{2log(1/\delta)} > \gamma$. Therefore, with  $\delta < e^{\frac{-\gamma^2}{2}}$, $H$ will be sampled again as $t\rightarrow \infty$ .

We now show this condition is strong enough to ensure continued sampling of $H$ until $V_{i,H}(s_t) > V_{i,L}(s_t)$.  We know any further samples of $H$ will result in the outcome $o_t = (H,H)$. Thus, at any time $t \geq 2$  $s_{t,2} = 1$ , therefore we now have the general upper confidence bound formulation of $H$ in this instance below:  
\[\text{UCB}_{i,H}( s_t) = \frac{s_{t,1}\beta}{s_{t,1} + 1} + \sqrt{\frac{2log(1/\delta)}{s_{t,1} + 1}}\]

Continued sampling of $L$ will strictly increase the value of $s_{t,4}$ and thus strictly decrease $\text{UCB}_{i,L}(s_t)$ down to $\gamma$. So, if agents reach some threshold where $V_{i,H}(s_t) > \gamma$ , then even with potentially more plays of $L$, there will eventually be some point $T$ where  $\text{UCB}_{i,H}(s_T) > V_{i,H}(s_T)>\text{UCB}_{i,L}(s_T)$ , so that $H$ will be the only action played for $t > T$. However, if $V_{i,H}(s_t) = \frac{s_{t,1}\beta}{s_{t,1} + 1} \leq \gamma$ , then this guarantee does not immediately hold. 

In this case, in order for agents to learn to collude, there must be enough plays of $H$ to increase its value such that $\frac{s_{t,1}\beta}{s_{t,1} + 1} > \gamma$.  However,  because continued play of $L$ will cause the values of $\text{UCB}_{i,L}(s_t)$ to approach $\gamma$, it is sufficient to find values of $\delta$ such that  $\text{UCB}_{i,H}(s_t) > \gamma$ at every value of $s_{t,1}$. As shown before, once $s_{t,1} > \frac{\gamma}{\beta - \gamma}$,  agents will eventually settle on colluding regardless of the value of $\delta$. So, we have the following inequality that must hold of all plays of $H$ (values of $s_{t,1}$):
\[\delta < \begin{cases}
exp(-(\gamma - \frac{s_{t,1}\beta}{s_{t,1} + 1})^2(\frac{s_{t,1}+1}{2})),&\text{if } s_{t,1} \leq \frac{\gamma}{\beta - \gamma}\\
1, & \text{otherwise}
\end{cases} .
\]

We see this bound is increasing in $s_{t,1} \geq 0$ on the domain $1>\beta>\gamma>0$. Therefore, the strictest upper bound on $\delta$ is that when $s_{t,1} = 0$. Thus, we see that this initial bound on $\delta$ that guarantees $H$ will be played a second time is sufficient to guarantee it will then be played a third time and so on. Thus, we see the upper bound on $\delta$ as $\delta < e^{\frac{-\gamma^2}{2}}$.
\end{proof}

\end{document}